\documentclass[lettersize,journal]{IEEEtran}
\usepackage{amsmath,amsfonts}
\usepackage{algorithmic}
\usepackage{algorithm}
\usepackage{array}
\usepackage[caption=false,font=normalsize,labelfont=sf,textfont=sf]{subfig}
\usepackage{textcomp}
\usepackage{stfloats}
\usepackage{url}
\usepackage{verbatim}
\usepackage{graphicx}
\usepackage{cite}
\usepackage{multirow}
\usepackage{xcolor}  
\usepackage{amssymb}   % 或 mathtools
\usepackage{booktabs}  % 导言区加入
\usepackage{pifont}

\usepackage{amsthm}

\newtheorem{definition}{\bf Definition}
\newtheorem{proposition}{\bf Proposition}
\newtheorem{theorem}{\bf Theorem}
\newtheorem{lemma}{\bf Lemma}
\newtheorem{corollary}{\bf Corollary}
\newtheorem{remark}{\bf Remark}

\begin{document}

\title{On Capacity and Delay of Wireless Networks with Node Failures}
\author{
  \IEEEauthorblockN{Wei Li, Min Sheng, \emph{Fellow, IEEE},  Junyu Liu,        \emph{Member, IEEE} and Jiandong Li, \emph{Fellow, IEEE}}
%  \IEEEauthorblockA{Department of Electrical Engineering 
%                    University 1
%                    City 1
%                    Email: author1@university1.edu}
%  \and
%  \IEEEauthorblockN{}
\\
  \IEEEauthorblockA{ State Key Laboratory of Integrated Service Networks, Xidian University, Xi’an, Shaanxi, 710071, China
 \\Email: xidianliwei@stu.xidian.edu.cn, \{junyuliu, yangzheng\}@xidian.edu.cn, \{msheng,jdli\}@mail.xidian.edu.cn
               }
}

        % <-this % stops a space
%\thanks{This paper was produced by the IEEE Publication Technology Group. They are in Piscataway, NJ.}% <-this % stops a space
%\thanks{Manuscript received April 19, 2021; revised August 16, 2021.}
%}

%%\markboth{Journal of \LaTeX\ Class Files,~Vol.~14, No.~8, August~2021}%


\maketitle

\begin{abstract}
One key challenge in designing resilient large-scale wireless ad hoc networks is to understand how random node failures affect fundamental network performance. In this work, we show that both network capacity and delay scale as \scalebox{0.65}{$\textstyle \Theta\left(\sqrt{\frac{n(1-q)}{\log n}}\right)$}, where $n$ is the total number of nodes and $q$ is the node failure probability.  The network capacity degenerates to the classical  result given by P. Gupta and P. R. Kumar when $q=0$. Based on these results, we find that even with the same number of non-faulty nodes, a network with $n$ nodes and node failure probability $q$ has lower network capacity than a failure-free network with $n(1-q)$ nodes. To compensate for the network capacity loss caused by random node failures, at least $\epsilon(n,q) nq$ redundant nodes are required, where $\epsilon(n,q)>1$.  We further prove that the optimal trade-off between network capacity and delay remains  $O(1)$ regardless of node failures, implying that high network capacity and low delay  cannot be achieved simultaneously. These results demonstrate robustness against stochastic variations in wireless channels.
\end{abstract}

\begin{IEEEkeywords}
Network capacity,  delay,  node failures, redundancy node, network topology.
\end{IEEEkeywords}

\section{Introduction}
Wireless ad hoc networks with self-configuration, self-optimization, and self-healing capabilities are crucial in emergency response, disaster rescue and military applications, etc\cite{ref2,Y1}. In  wireless ad hoc networks, data packets are transported with high probability through the multi-hop transmission strategy, which effectively mitigates intra-network interference  and enhances spatial reuse ratio by avoiding long-distance transmissions\cite{ref13}.  One of the necessary conditions for implementing the multi-hop transmission is that the network topology should be connected with probability one, i.e.,  there is at least one  path  between any  two   nodes.  In practice, however,  the network topology may experience intermittent connectivity due to deteriorating channel conditions, malicious interference,  node failures, node mobility, and other factors \cite{ref9,ref10,ref12}. Among these factors, the impact of node failures on the connectivity of network topology is irreversible, which can deteriorate  the network traffic-carrying capacity and increase delay.    Based on the above analysis, modeling and quantifying the impact of node failures on network capacity, delay, and the trade-off among them is the cornerstone of designing a wireless ad hoc network with  resilience.  

Accurately characterizing the network capacity region and delay of wireless ad hoc networks is  challenging, espeically when  the number of nodes increases.  As an alternative, analyzing the asymptotic scaling behavior of network capacity and delay can also provide insights into the fundamental performance limit of wireless ad hoc networks. Specifically, given that an optimal space-time  scheduling strategy is provided, the network capacity  is $\Theta \left(\sqrt{n}\right)$ bits/s, where $n$ is the number of nodes \cite{ref13}. This means that the rate at which all source nodes can successfully transmit data packets to their corresponding destination nodes in the same time slot is $\Theta \left(1/\sqrt{n}\right)$ bits/s. Compared to the direct transmission strategy, the data rate between each source-destination pair can be improved from $\Theta \left(1/n\right)$ to $\Theta \left(1/\sqrt{n}\right)$ with the help of the multi-hop  transmission strategy.  Furthermore, the mobility-assisted two-hop scheduling strategy can achieve linear growth in network capacity with respect to the number of nodes, i.e.,  $\Theta\left(n\right)$,  which would yet result in  intolerable delay \cite{CD0}. This is undesirable for most applications.   Following  these works, a number of researchers  have conducted in-depth studies on the relationship between network capacity and delay from various perspectives.  The main results are summarized in Table \ref{Tab1}. The core   is to investigate the asymptotic scaling  behavior of network capacity and delay under various networking schemes. Compared to the results in Table \ref{Tab1}, the networking scheme is the key factor that constrains the asymptotic behavior of network capacity and delay. Designers can enhance network capacity and reduce delay by designing networking schemes that incorporate new dimensions,  such as node mobility, redundancy, coding and geometric information, etc. All the aforementioned  studies are conducted under the assumption that nodes could always operate well. The results fail to reveal the impact of node failures on network capacity and delay, thereby hindering the design of wireless ad hoc networks with resilience.  Recently, the  network capacity of  wireless ad hoc networks after experiencing zone node failures has been studied in \cite{ref14,ref15,ref24R}. It is demonstrated that increasing the connectivity of  network  topology to combat node failures is feasible  only when the number of failure regions is less than $\Theta\left(\sqrt{n}\right )$. It is also pointed out  that the order of network capacity, when the network encounters node failures, follows the same asymptotic scaling law as in \cite{ref13}. From the statistical average perspective, the result obtained using the proof method in \cite{ref13} is correct. However, those results do not capture the impact of random node failures on the network connectivity  and  network capacity. To this end, we will quantify this subtle difference from the perspective of network connectivity.   We show that this overestimates the network capacity when the deployed wireless ad hoc networks encounter random node failures. Meanwhile, the delay and the optimal trade-off between network capacity and delay are given.

\begin{table}[!t]
\caption{Network Capacity and Delay in Different Networking Schemes}
\label{Tab1}
\centering
\small
\setlength{\tabcolsep}{4pt}   % 列间距稍微加大
\renewcommand{\arraystretch}{1.75} % 行距变宽（关键）

\begin{tabular}{lccc}
\toprule
\textbf{Networking Scheme} & \textbf{Capacity} & \textbf{Delay} & \textbf{Ref.} \\
\midrule

Multi-hop + Mobility 
& $\Theta\!\left(\frac{\min\{m,n\}}{\log^{3}n}\right)$ 
& $\Theta\!\left(\frac{1}{v}\right)$ 
& \cite{CD1} \\[2pt]

Multi-hop + Redundant  
& $\Theta\!\left(1/\log n\right)$ 
& $\Theta(\log n)$ 
& \cite{CD2} \\[2pt]

Two-hop 
& $\Theta(\sqrt{n})$ 
& $\begin{array}{c}
\Theta(n^{2\beta}) \\
{\scriptstyle 0 \leq \beta \leq 0.5}
\end{array}$ 
& \cite{CD3} \\[4pt]

Multi-hop + Coding 
& $\Theta(\sqrt{Dn})$ 
& $n^{1/3} < D < n$ 
& \cite{CD4} \\[2pt]

Multi-hop    
& $\Theta(\sqrt{n})$ 
& $\Theta(\sqrt{n})$ 
& \cite{ref19} \\[2pt]

Geometric multi-hop    
& $\Theta\!\left(\frac{n}{\log n \log\log n}\right)$ 
& $\Theta\!\left(\frac{1}{v}\right)$ 
& \cite{CD5} \\

\bottomrule
\end{tabular}
\end{table}

 The contributions are summarized as follows.
\begin{itemize}
\item We derive the asymptotic scaling laws for both network capacity and delay  in wireless ad hoc networks with node failure probability $q$  as \scalebox{0.5}{$ \Theta \left( \sqrt{\frac{n(1-q)}{\log n}}\right)$}. The upper bound on network capacity can be achieved by utilizing multi-hop  transmission strategy.  When $q=0$, our results degenerate into the  classical results in \cite{ref13}. Based on  the above conclusions, increasing redundant  nodes  is a feasible strategy for both restoring the connectivity of network topology and mitigating node failure randomness. In particular, to maintain the network capacity at \scalebox{0.5}{$\Theta\left(\sqrt{\frac{n}{\log {n}}}\right)$} in  wireless ad hoc networks with node failures, at least $n_{1} = \epsilon(n,q) nq$ redundant nodes should be deployed, where $\epsilon(n,q)>1$. 
\item It is shown that   the wireless ad hoc networks with $n$ nodes and node failure probability $q$ have lower network capacity than one with $n(1-q)$ nodes without considering node failures, even though both networks contain the same number of non-faulty nodes. For the case with node failures, the loss of network
capacity is utilized to mitigate the impact of node failure randomness  on the connectivity of  network topology, as the expected number of non-faulty nodes is $n(1-q)$. The numerical results indicate that when $q=0.85$, the network capacity loss rates corresponding to $n=$ 100, 1000, and 10,000 are 24\%, 15\%, and 11\% respectively, compared to scenarios where node failures are not considered.

\item The ratio of network capacity to delay is $O(1)$, which is also order-optimal in the presence of node failures. This reveals that enhancing network capacity  comes at the cost of increased delay and  this relationship is independent of the node failure probability and spatial dimension.  In other words,   it is not possible to simultaneously achieve high network capacity and low delay in multi-hop wireless ad hoc networks.   For scenarios with lower delay,  using one hop or fewer for data transmission may be a preferable strategy. Conversely, for scenarios with high traffic load, the multi-hop  transmission is the optimal strategy. 
\end{itemize}

Notations: $f(n)=O(g(n))$ if there exist $c>0$ and $n_{0}>0$ such that $f(n)<cg(n)$ for $n>n_{0}$. $f(n)=\Omega(g(n))$ if $g(n)=O(f(n))$.
$f(n)=\Theta(g(n))$ if $f(n)=O(g(n))$ and $g(n)=O(f(n))$.  $\rho(x,y)$ is the Euclidean distance between positions $x$ and $y$. $Card(A)$ denotes the cardinality of the set $A$.

\section{System model and definitions}
We consider a wireless ad hoc network consisting of 
$n$ nodes. All the nodes are uniformly and independently distributed in unit square $\mathcal{A}=[0,1]^{2}$. The coordinate of node $i$ is denoted as $x_{i}$, where
$i\in I=\{1,2,3,\cdot\cdot\cdot, n\}$. All nodes use the same  transmission power $P$,  i.e.,  they have the same transmission radius, which  is a function of the number of nodes.  Each node is equipped with either an omnidirectional or a directional  antenna based on  scheduling constraints and practical requirements. All transmissions use the full bandwidth $W$.  $N(n)=\frac{n}{2}$ unicast data streams are generated by using a random, uniformly matching method. Each node serves as either a source or a destination for one  data stream and also relays data packets from other data streams. Time is slotted for packetized transmission. In most mobility scenarios, moreover, nodes move much more slowly than the time it takes to transmit a data packet. Therefore, it is reasonable to model mobile networks as static when analyzing capacity.

\subsection{Network topology}

Due to the factors  including physical disruptions and both external and internal  interference, etc., the state of the nodes becomes highly unreliable, which results in  intermittent connectivity in the network topology.  To describe the random failure behavior of nodes, we define the node failure probability, denoted as $q$. Let $\chi$ denote the set of non-faulty nodes   and $\chi^{c}$ denote the set of faulty nodes. $Card(\chi)+Card(\chi^{c})=n$, where $Card(\chi)=(1-q)n$. The non-faulty nodes can form a random geometric graph  $G(n,r_{q}(n))$, where $r_{q}(n)$ is the critical transmission radius, which is the minimum transmission radius that ensures $G(n,r_{q}(n))$ is connected with  probability 1. In $G(n,r_{q}(n))$, the rule for establishing a link between two nodes is given by
\begin{equation}
    w(x_{i},x_{j}) = \begin{cases} 
1, & \: \rho(x_{i},x_{j})\leq r_{q}(n), \\\\
0. & \text{otherwise}.
\end{cases}
\label{Connec0121}
\end{equation}
where $\rho(x,y)$ is the Euclidean distance between point  $x$ and point  $y$. When wireless channel conditions are taken into account, a link between node \(x_{i}\) and node \(x_{j}\) is established with a certain probability. This implies that nodes separated by a distance greater than \(r_{q}(n)\) may also form a link. Such minor random fluctuations do not fundamentally affect the scaling behavior of the network capacity (see Appendix \ref{AppendixA1}).

$G(n,r_{q}(n))$ is  connected if and only if there exists a path between any pair of nodes. In other words, the random geometric graph contains no isolated nodes with probability 1 as $n \to +\infty$. The definition of the connectivity of  network  topology is given from \cite{ref16} as follows.

\begin{definition}The random geometric graph $G(n,r_{q}(n))$ formed by all non-faulty nodes  is connected, if  

   a) For any $x_{i},x_{j}\in\chi, i\neq j$, there exists at least one  data path between the two nodes;

    b) For any $x_{t}\in\chi^{c}$, $N^{+}(x_{t})\cap \chi\neq\emptyset$, where $N^{+}(x_{t})=\{x_{i}\:|\:\rho(x_{i},x_{t})\leq r_{q}(n),\: x_{i}\in \chi \}$ is the set of neighbor nodes of the faulty node $x_{t}$.  
\end{definition}

 After the deployment of a wireless ad hoc network, node failures occur randomly and instantaneously. Therefore, to ensure network connectivity, condition b) in Definition 1 requires that each failed node be located within the transmission range of at least one non-faulty node. In addition, condition b) effectively ensures that a failed node once recovered to a non-faulty state can rejoin the network structure formed by the non failed nodes with  probability one.   We conducted simulations on the connectivity of  network structure  for different node failure probabilities as the node transmission radius varies, as shown in Fig. \ref{Fig2}. The network  topology is connected with probability 1 when the transmission radius is greater than the  critical transmission radius. Otherwise, the network topology is not connected with probability one. 
 
 \begin{figure}[!t]
\centering
\includegraphics[scale=0.525]{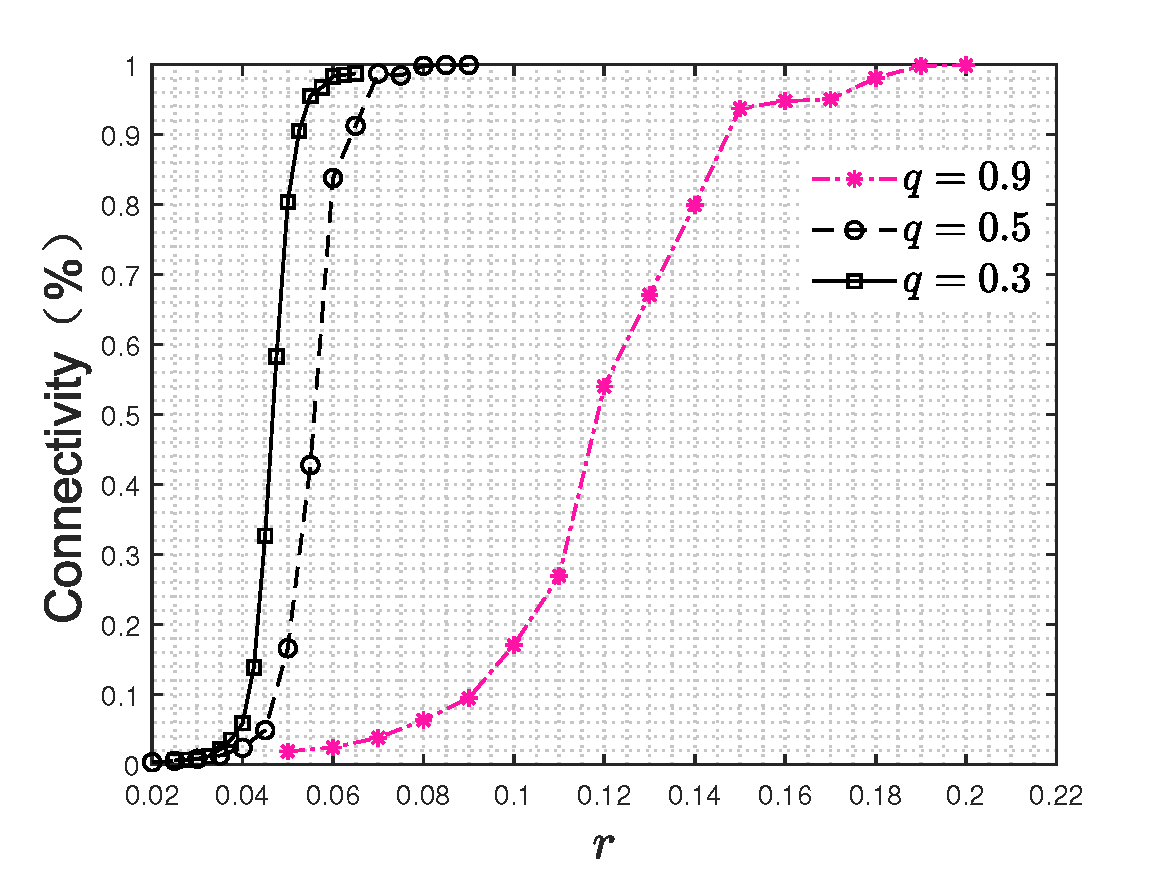}
\caption{The connectivity of communication  topologies varies with the transmission radius under different the  node failure probabilities---$q=.9,.5,.3$  ($n=1000$, each sampling point is simulated 10,000 times)}
\label{Fig2}
\end{figure}

\subsection{Geometric conditions for successful transmission\label{s221}}
Using the protocol model,  we model the single-hop successful transmission conditions for nodes operating with different antenna types \cite{ref18,ref181}. Let $\mathcal{H}=\left\{(i,R(i))\:|\:i\in J\right\}$ denote the set of transmitter-receiver pairs in the same time slot, where $J$ is the set of indices of active nodes  and $R(i)$ is the receiving node corresponding to node $i$.  The size of the set $\mathcal{H}$ is determined by the spatial distribution of nodes and the traffic pattern. 

\begin{definition} A data packet  is successfully transmitted in one hop if either of the following conditions is satisfied.

\textbf{C1: Omnidirectional antenna.} Node $x_{i}$ successfully transmits a data packet to node $x_{R(i)}$ if $\rho(x_{i},x_{R(i)})<r_{q}(n)$ and the distance between  transmitting nodes $x_{k}$ and $x_{R{i}}$ satisfies  

\begin{equation}
\rho(x_{k},x_{R(i)})\geq (1+\triangle)\rho(x_{i},x_{R(i)}), \label{Eq-1}
\end{equation}
where $k\in J-\{i\}$ \footnote{$J - \{x\}$ is defined as the set of all $y$ such that $y$ is an element of $J$ and $y$ is not an element of $\{x\}$, i.e.,
\(J - \{x\} = \{\, y \mid y \in J \text{ and } y \notin \{x\} \,\}.\)} and $\triangle$ is the in-network interference parameter.

\textbf{C2: Directional antenna.} Under the constraint of a minimum beamwidth $\theta$, node $x_{i}$ successfully transmits a data packet to node $x_{R(i)}$ if $\rho(x_{i},x_{R(i)})\leq r_{q}(n)$ and node $x_{R(i)}$ is either outside the beam of simultaneously transmitting nodes $x_{k}$ or equation (\ref{Eq-1}) holds.
\label{DP2}
\end{definition}

\begin{figure}[!t]
\centering
\subfloat[Omnidirectional]{%
    \includegraphics[width=0.45\linewidth]{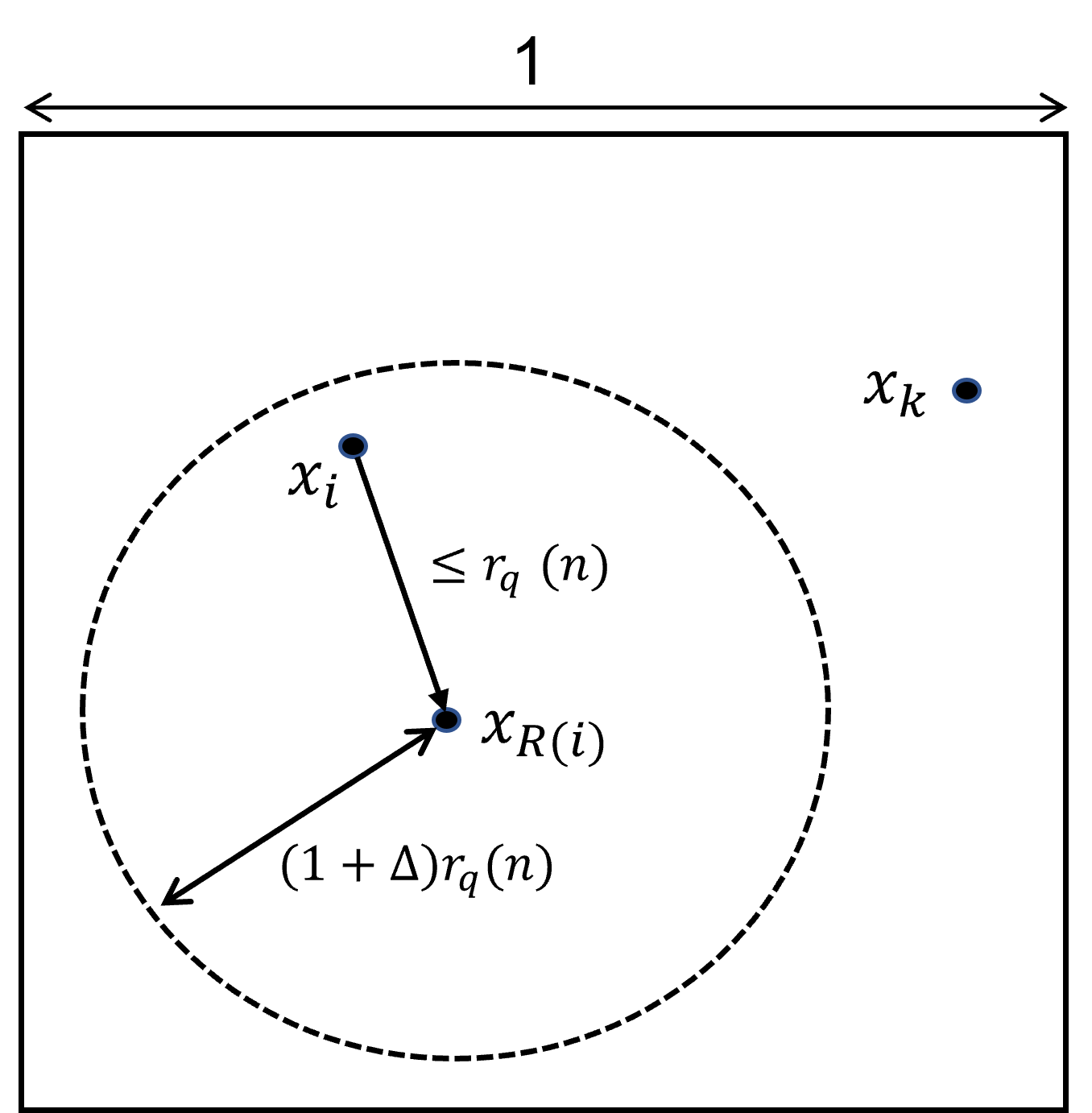}
}
\hfill
\subfloat[Directional]{%
    \includegraphics[width=0.45\linewidth]{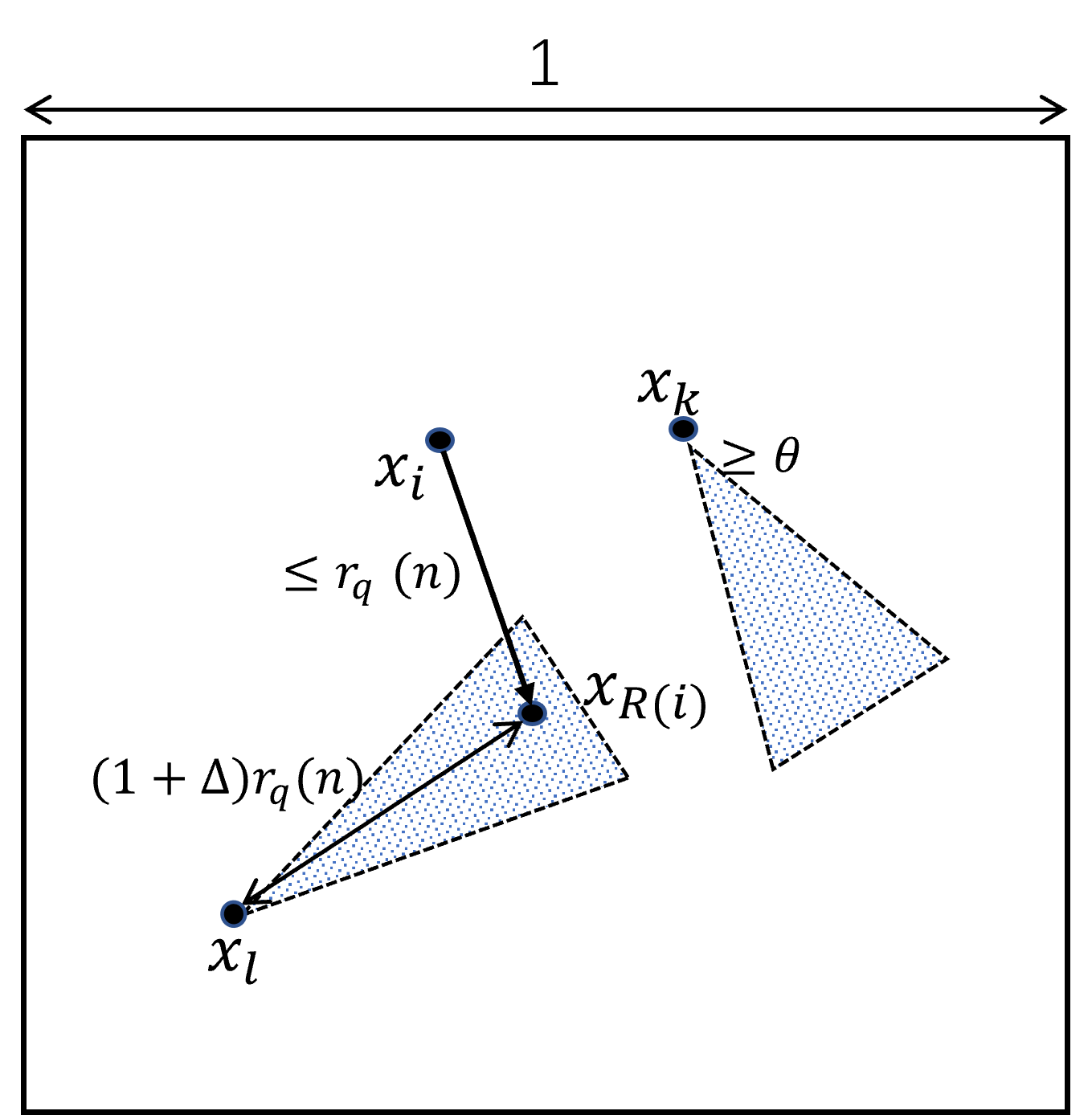}
}
\caption{Geometric conditions of single-hop successful transmission}
\label{Fig3}
\end{figure}

To illustrate the protocol model, the geometric interpretations of  C1 and C2 are given in Fig. \ref{Fig3}. Combining Definition \ref{DP2} and the triangle inequality, we know that, for multiple transmitter-receiver pairs to simultaneously succeed in transmitting data packets, the disks centered at the receiver node $x_{R(i)}$ with radii $\frac{\Delta^{\prime}}{2}\rho(x_{i},x_{R(i)})$ should  not overlap (see Fig. \ref{Fig1}). This geometric condition holds for both omnidirectional and directional antennas \cite{ref181}. For the unified analysis, we use the in-network interference parameter $\Delta^{\prime}$, which is given by  

\begin{eqnarray}
     \Delta^{\prime}=
    \begin{cases}
    \Delta,& \mbox{omnidirectional,}\\\\
    \min\left\{\Delta,\sin \frac{\theta}{2}\right\}, &\mbox{directional.}
    \end{cases}
    \label{eq3}
\end{eqnarray}

According to  (\ref{eq3}),  the impact of different antenna types on the in-network interference parameter 
is a constant and  is independent of the number of nodes. Therefore,  we do not  distinguish the antenna types in  analyses.

In general, the signal-to-interference-plus-noise ratio (SINR) is also used to determine whether a communication link is successfully established between two nodes. This is  called the physical model.
 Specifically, node \( R(i) \) can successfully receive a data packet from node \( i \) within time slot \( t \) if the following condition is satisfied:

\begin{equation}
\frac{P_{i}\ell(x_{i},x_{R(i)})}{N + \sum_{\substack{k \in \mathcal{H} \\ k \neq i}} P_{k} \ell(x_{k},x_{R(i)})} \geq \beta
\label{SINR1}
\end{equation}
where \( N \) denotes the additive noise power level and \( \beta \) is the SINR  threshold required for successful decoding, which depends on the transmitter hardware and the signal strength level rather than being a constant value.  \(\ell(x, y) = \left( \rho(x, y) \right)^{-\alpha}\)  is the path loss function, where  \( \rho(x, y) \geq \rho_0 > 0 \), and \( \alpha \:\:(>2)\) is the path loss exponent.
\begin{figure}[!t]
\centering
\includegraphics[width=5.5cm,height=5.5cm]{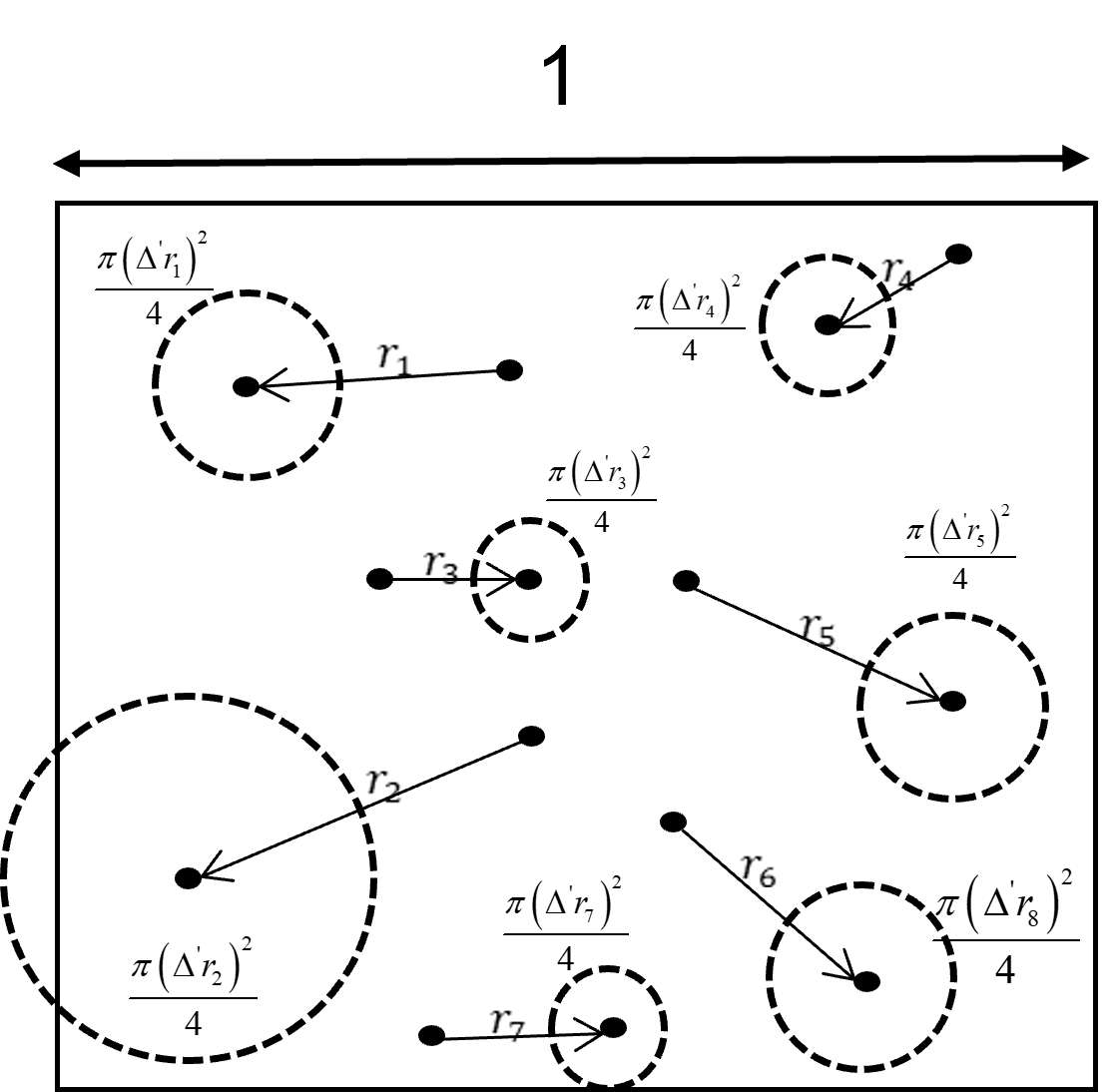}
\caption{ Many nodes successfully transmit data packets simultaneously ($r_{i}\leq r_{q}(n), i=1,\cdot\cdot\cdot,7$)}
\label{Fig1}
\end{figure}

In comparison with Definition~\ref{DP2}, the key parameters in the physical model defined by (\ref{SINR1}) are the demodulation threshold \( \beta \) and the path loss exponent \( \alpha \), whereas those in the protocol model is the in-network interference parameter \( \Delta^{\prime} \).
 These key parameters essentially measure the impact of in-network interference caused by simultaneous transmissions on  links. These two models are equivalent, as shown in Remark \ref{Remarl1017}.

\begin{remark}
According to Theorem 4.1 in~\cite{ref18}, if \( \Delta^{\prime} > \Delta(\beta) := \left(48 \beta \frac{2^{\alpha - 2}}{\alpha - 2}  \right)^{\alpha^{-1}} \), then any set of simultaneous  transmitter-receiver
pairs \( \mathcal{H}\) allowed under the protocol model can also be supported under the physical model with SINR threshold \( \beta \), given a suitable power assignment \( \{P_i, 1 \leq i \leq n\} \). This means that the capacity obtained using the protocol model is robust to both in-network interference and path loss. As the path loss exponent increases, the value of $\Delta(\beta)$ decreases, and the area consumed by successful link establishment becomes smaller.
\label{Remarl1017}
\end{remark}

\subsection{Network protocols\label{S23}}
 The key issues in using multi-hop  transmission scheme are selecting the data paths for transporting data packets and avoiding in-network interference caused by multiple nodes transmitting in the same time slot. The former is the routing problem and the latter is the link scheduling problem. Since we aim to focus on the fundamental performance limit of wireless ad hoc networks, we outline the basic principles of multi-hop  transmission scheme. The standard proof framework in \cite{ref112} is employed to establish the relationship between key parameters and network capacity.

\subsubsection{Mapping relationship between site percolation graph and random geometric graph}
To extract the critical parameters ensuring the network  topology is connected with probability 1, we tile the unit square region with small squares of side length $a_{q}(n)\triangleq \frac{r_{q}(n)}{\sqrt 2}$, as shown in Fig. \ref{Fig5}. On one hand, this method ensures that each small square contains at least one node, which indicates that  each small square has the capability to relay data packets from other nodes. On the other hand, it transforms a random geometric graph formed by non-faulty nodes into a site percolation graph, where  each vertex has at most four neighbors.

Each small square can be equivalently considered as a vertex in the site percolation graph.  if each of two adjacent small squares contains at least one non-faulty node,  there exists an edge between the corresponding vertices. The failure probability  for non-empty small squares is $q^{\prime}$, which is determined by the number of non-faulty nodes  within each small square.   When $q^{\prime}$ is less than the critical probability $q_{c}$, the connectivity of the site percolation graph increases sharply to one with probability 1 (forming a  unique connected component). Conversely, the site percolation graph is disconnected with probability 1 (forming multiple connected components). In discrete percolation theory, this sharp change in connectivity is known as  the phase transition and its strict definition is given as follows \cite{ref20}.

\begin{definition}The probability that the connectivity of the site percolation graph approaches one is denoted as $\psi(q^{\prime})$. The condition is given by
\begin{eqnarray}
     \psi(q^{\prime})=
    \begin{cases}
    1,& q^{\prime}\leq q_{c},\\
    0, &q^{\prime}>q_{c}.
    \end{cases}.
    \label{eq4}
\end{eqnarray}
where $q_{c}$ is the critical failure probability\footnote{Through simulation in \cite{ref20}, it is found that $q_{c}\approx 0.4073.$}.
\label{D2}
\end{definition}

 \begin{figure}[]
\centering
\includegraphics[width=8cm,height=4.5cm]{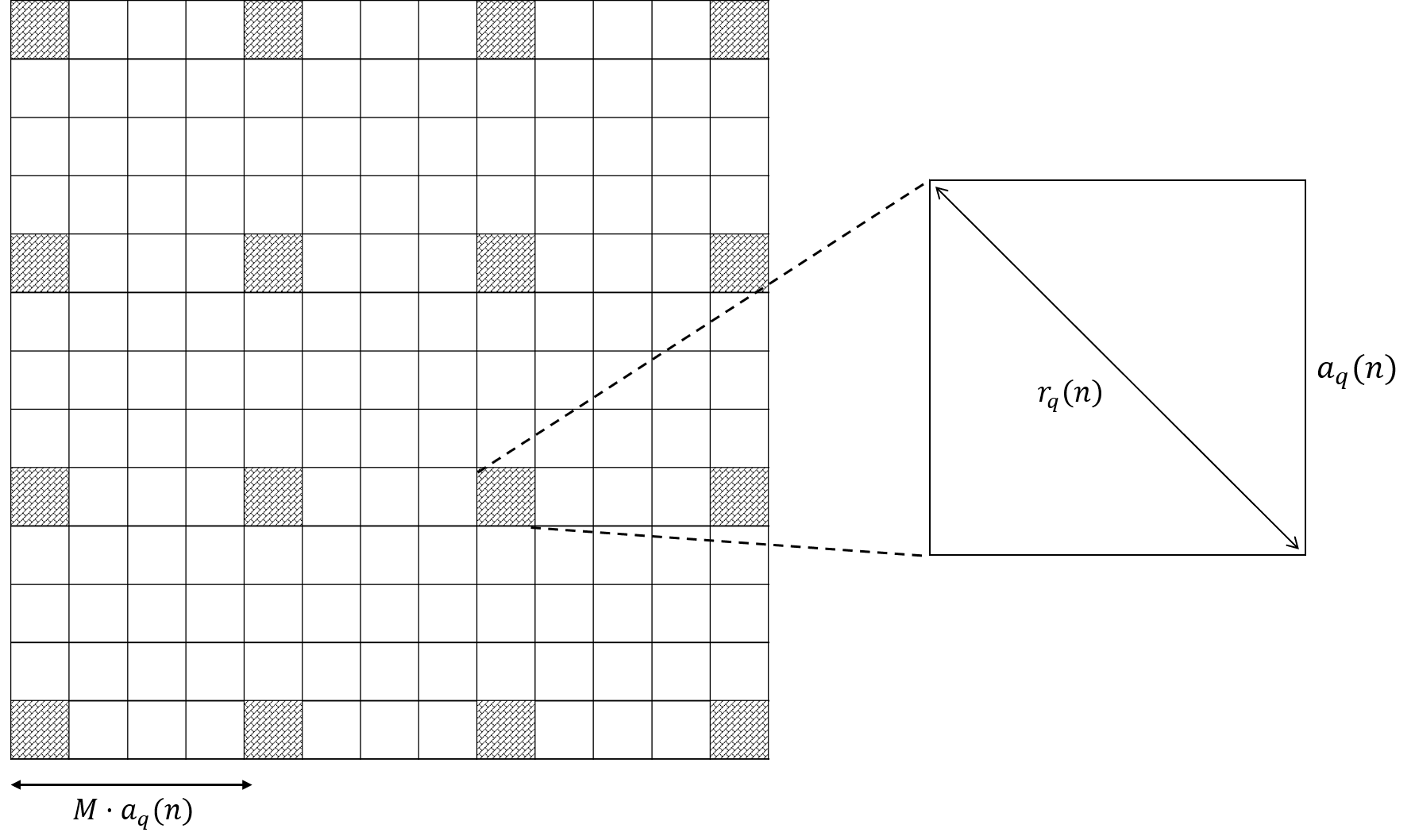}
\caption{The division mode of unit squares}
\label{Fig5}
\end{figure}

\subsubsection{Interference-free scheduling} To enhance spatial reuse and minimize in-network interference, interference-free  scheduling  is employed. We divide all small squares into clusters, where each cluster contains $M^{2}$ small squares, as shown in Fig. \ref{Fig5}. Small squares within each cluster are numbered from left to right and from top to bottom.  During the scheduling process, all small squares within each cluster that share the same number are activated simultaneously, and then the non-faulty nodes within each small square transmit data packets in sequence. Interference-free scheduling  ensures that multiple non-faulty nodes can transmit successfully in the same time slot, thereby increasing spatial reuse. 

Combining the protocol model  in Definition \ref{DP2}, $M$  is given by 

\begin{equation}
\begin{aligned}
M &\geq \left\lceil 1+\frac{r_{q}(n)+(1+\Delta^{\prime})r_{q}(n)}{r_{q}(n)/\sqrt{2}} \right\rceil \\
  &= \left\lceil 1+\sqrt{2}(2+\Delta^{\prime})\right\rceil \geq 3,
\end{aligned}
\label{M7}
\end{equation}
where $\lceil x \rceil$ denotes the smallest integer greater than $x$. From equation (\ref{M7}), it is clear that  $M$ is determined only on $\Delta^{\prime}$.  From Remark \ref{Remarl1017}, the value of $\Delta^{\prime}$ is determined solely by $\beta$ and $\alpha$, independent of the network size.  Therefore, the value of $\Delta^{\prime}$  does not affect the order of network capacity and delay.

\subsubsection{Multi-hop  transmission strategy}
 A multi-hop  transmission strategy is used to transport data packets between a source node and a destination node. For each source-destination pair $i$, the corresponding S-D line is constructed by connecting the source node $S_{i}$ and the destination node $R_{i}$. Subsequently, data packets generated by the source node  $S_{i}$ are relayed through the small squares along the S-D line to the destination node $R_{i}$, as shown in Fig. \ref{Fig7}(a). If the relaying process encounters an empty small square, a rerouting strategy is employed, as shown in Fig. \ref{Fig7}(b). There are two cases that can lead to a small square being empty: i) it contains no nodes at all; ii) it contains  nodes, but failures occur for all the nodes. Assume that the coding and queue scheduling algorithms  at each hop are optimal and independent of the number of nodes.

\subsection{Definitions of Network Capacity and Delay}

\begin{figure}[]
\centering
\subfloat[Routing strategy]{%
    \includegraphics[width=0.45\linewidth,height=3.5cm]{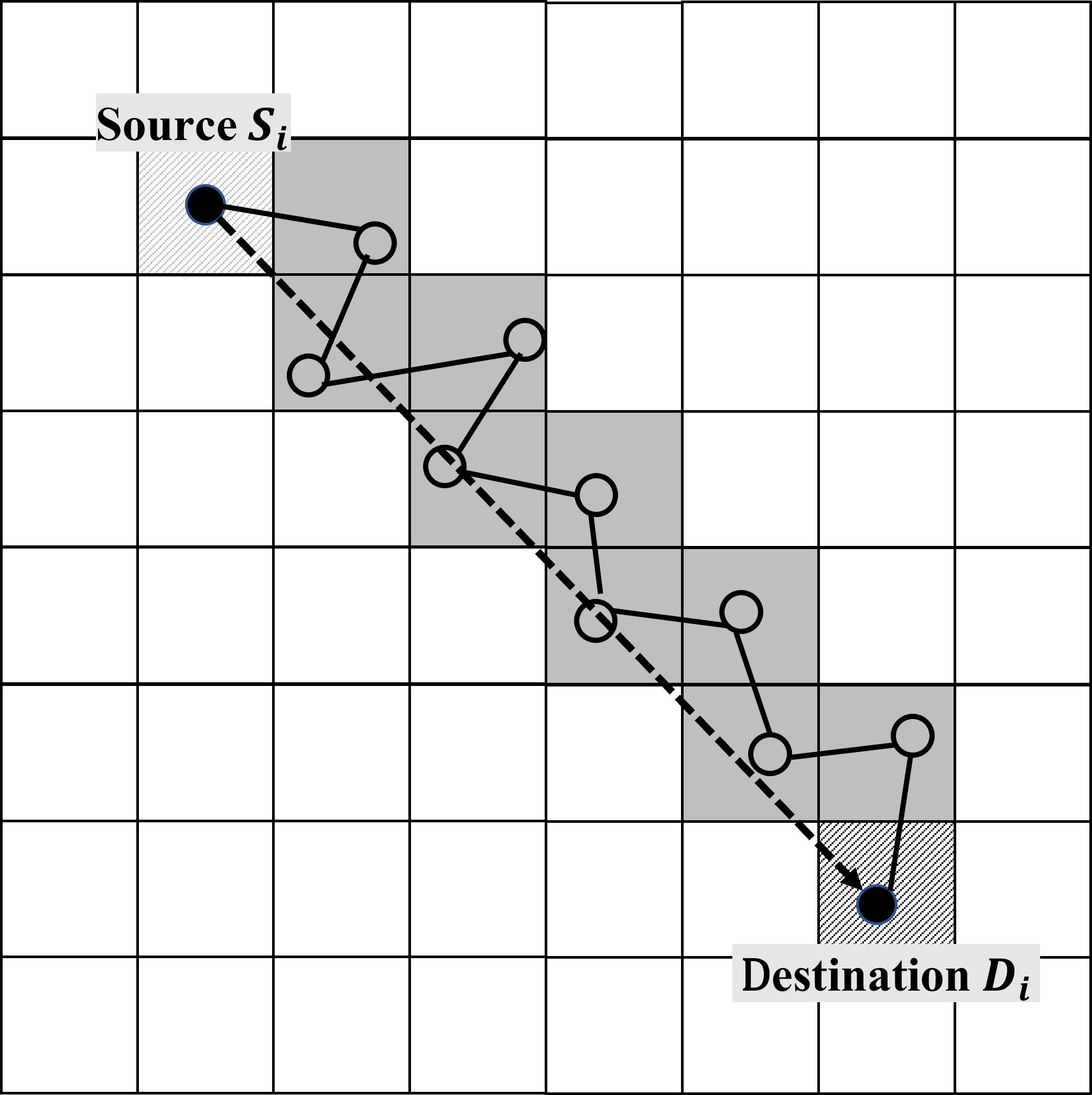}
}
\hfill
\subfloat[Rerouting strategy]{%
    \includegraphics[width=0.45\linewidth,height=3.5cm]{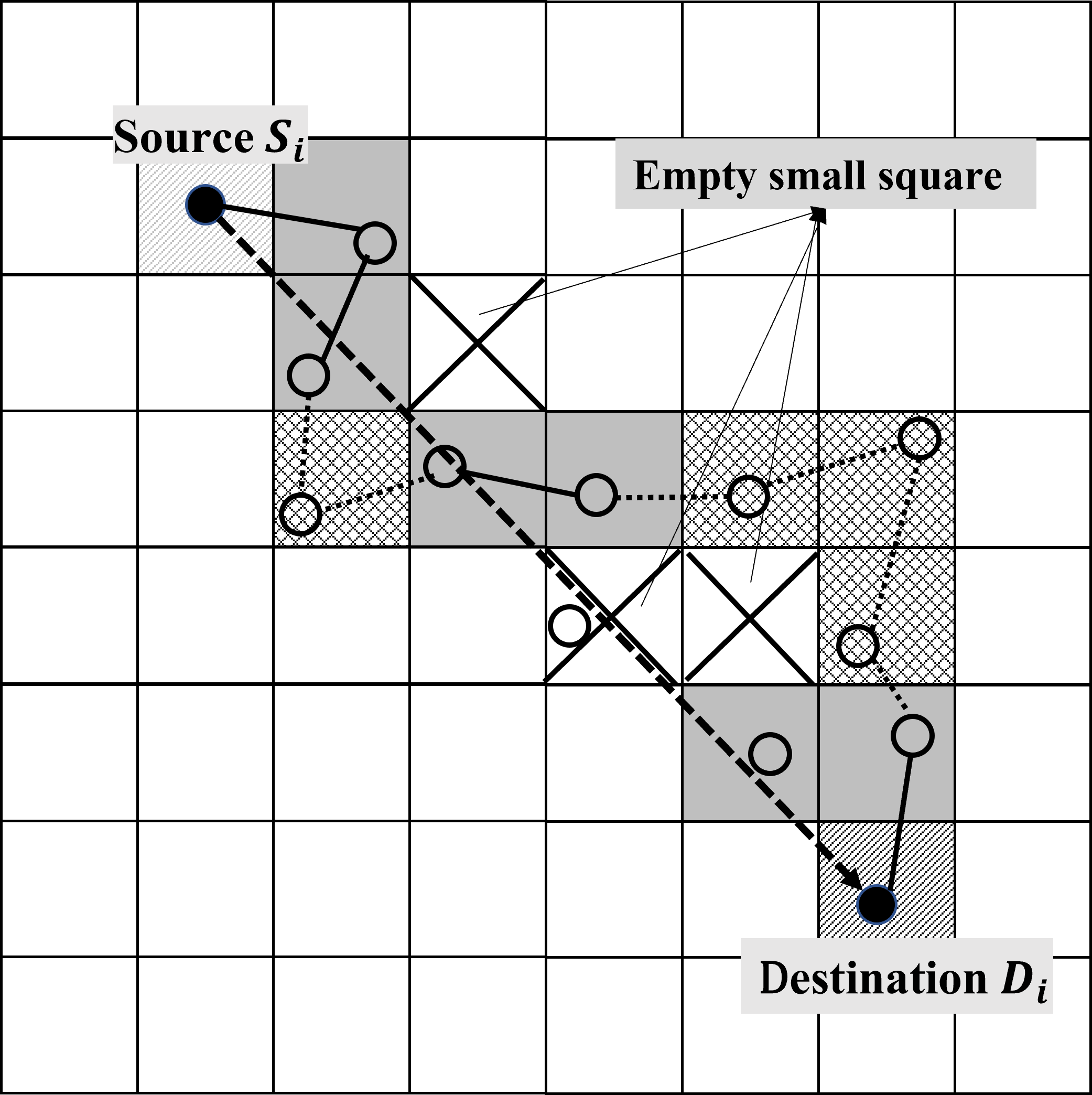}
}
\caption{Multi-hop transmission strategy}
\label{Fig7}
\end{figure}

\begin{definition} Network capacity  is defined as 
$$S_{q} (n)=N_{q}(n)\lambda_{q}(n),$$
where $N_{q}(n)$ is the number of source-destination pairs.  The feasible data transmission rate $\lambda_{q}(n)$ is  defined as follows \cite{ref13}.

$$\lambda_{q}(n)\triangleq \underset{1\leq i\leq N_{q}(n)}{\min}\lim\limits_{t\to\infty}\inf \frac{B(i,t)}{t}\:\: \text{bits/s},$$ 
 where $B(i,t)$ is the cumulative number of bits successfully transmitted for  source-destination pair $i$. $q$ is the node failure probability.

\label{D3}
\end{definition}

\begin{definition} Delay is defined as

$$D_{q}(n)=\frac{1}{N_{q}(n)}\sum_{i=1}^{N_{q}(n)}E\left[\overline{D}^{i}(n)\right],$$
where $E\left[\overline{D}^{i}(n)\right]$,    the average delay for transmitting data packets for source-destination pair $i$,   is defined as follows \cite{ref19}.
$$E\left[\overline{D}^{i}(n)\right]\triangleq E\left[\lim\limits_{k\to\infty}\sup \frac{1}{k}\sum_{j=1}^{k}D_{j}^{i}(n)\right],$$ 
where $D_{j}^{i}(n)$  is the total time required for the the $j$-th data packet of source-destination pair $i$ to travel from the source node to the destination node. $q$ is the node failure probability.

\label{D4}
\end{definition}

\section{Network capacity and delay \label{Set3}}
Before analyzing network capacity and delay, we first give the critical transmission radius and  the number of non-faulty nodes in each small square.
\begin{lemma}When the node failure probability is $q$, the critical transmission radius that ensures the network topology remains connected with probability 1 is given by \scalebox{0.75}{$r_{q}(n)=\sqrt{\frac{\log {n}+\xi}{ (1-q)n}}$},  where  $\xi >0$. When node failure is not considered, the critical transmission radius is \scalebox{0.75}{$r(n)=\sqrt{\frac{\log {n}+\xi^{\prime}}{ n}}$}, where $\xi^{\prime}>0$.

\label{LL1}
\end{lemma}

\begin{proof}
This proof is provided in \cite{ref16}.
\end{proof}
\begin{remark}In real communication scenarios, wireless channel conditions such as fading, shadowing, and path loss affect instantaneous link quality. However, prior studies (e.g., \cite{channel,channe2}) show that wireless network connectivity is generally not sensitive to these factors. Specifically, shadowing can enhance connectivity by increasing the likelihood of long-range links, whereas fading and path loss can slightly degrade it. These variations do not impact the scaling behavior of the critical transmission radius, as established in Lemma \ref{LL1}. Consequently, the capacity scaling laws derived in the subsequent analysis remain robust to the random fluctuations of wireless channel. Some simulation results on the impact of channel randomness on network structural connectivity are provided in Appendix \ref{AppendixA1}.
\end{remark}

 According to Definition \ref{D2},  if the failure probability of a non-empty small square  is below the critical failure probability, the site  percolation graph is connected with probability 1. Otherwise, it is disconnected with probability 1.  Whether a non-empty small square fails is determined by the number of non-faulty nodes it contains and the node failure probability. In Proposition \ref{P1}, we provide the number of non-faulty nodes in each small square.

\begin{proposition}
For $a_{q}(n)\geq \frac{r_{q} (n)}{\sqrt{2}}$,  the number of non-faulty nodes in each small square is $\Theta(\log (n))$.
\label{P1}
\end{proposition}

\begin{proof}The detailed proof can be found in Appendix  \ref{AppendixA}.
\end{proof}

Combining Proposition \ref{P1} and the conditions in Definition \ref{D2}, the condition for ensuring the site percolation graph is connected with probability 1 is
\begin{equation}
\begin{aligned}
q' = q^{c_{1}\log n} < q_c \Rightarrow q < q_c^{\frac{1}{c_{1}\log n}} 
\approx 0.4073^{\frac{1}{c_{1}\log n}}, 
\end{aligned}
\label{eq15}
\end{equation}
where  $c_{1} > 0$.

According to (\ref{eq15}), when the number of nodes is fixed, the node failure probability should be less than a certain critical value to ensure the site percolation graph is connected with probability 1.  This indicates that increasing the number of nodes can enhance the resilience of the network  topology. In the subsequent analysis, assume that condition (\ref{eq15}) and Proposition \ref{P1} both hold, meaning the site percolation graph is connected with probability 1.

\subsection{Network capacity\label{Se3_1}}

In this section, we will present the network capacity of wireless ad hoc networks with node failures.

\begin{theorem}The network capacity of wireless ad hoc networks   is 
$$S_{q}(n)=\Theta\left( \sqrt{\frac{n(1-q)}{\log n}}\right)\:\:\text{bits/s},$$
where $q$ is the node failure probability. As \( p \to 0 \), this result reduces to $\Theta\left(\sqrt{\frac{n}{\log {n}}}\right)\:\:\text{bits/s}$.
\label{Th1}
\end{theorem}
\begin{proof}
We are going to provide the proofs for the upper and lower bounds of  network capacity, respectively.

\ding{192}  Lower bound. 

The lower bound on network capacity ensures there exist interference-free scheduling and routing strategies enabling the network to achieve at least the given capacity.  According to Section \ref{S23}, specifically, interference-free scheduling  ensures that each small square within a cluster can only be activated once within $M^{2}$  time slots. This means that fewer small squares are activated in the same time slot when $M$ is larger, resulting in less interference. The degradation in network capacity due to interference-free scheduling is $\Theta\left(\frac{W}{M^{2}}\right)$.  During the multi-hop routing process,  nodes within each small square not only transmit and receive data packets corresponding to their own source nodes but also relay data packets from other source-destination pairs. Therefore, the total number of data packets forwarded  is proportional to the total number of S-D lines crossing the small square. This essentially limits the lower bound of the achievable network capacity, i.e., there exists a feasible multi-hop strategy that can achieve the corresponding lower bound of the network capacity.

Firstly, we provide the number of S-D lines served by each small square. Let $H_{i}$ be the total number of hops required from the source to the destination for pair $i$, $i=1,\cdot\cdot\cdot, N_{q}(n)$.   Since nodes are uniformly distributed, we have

 \begin{equation}
    E[H_{i}]=\Theta \left(\frac{E[L_{i}]}{a_{q}(n)}\right),
    \label{P21}
\end{equation}
where $E[L_{i}]$ be the distance between the source node and destination node of pair $i$.

 Let $Y_{i}^{j}$  be the indicator variable for whether the S-D line of source-destination pair $i$ passes through small square $j$. 
 $Y_{i}^{j}=1$  indicates that the S-D line of source-destination pair $i$ passes through small square $j$ and $ Y_{i}^{j}=0$  indicates that the S-D line of source-destination pair $i$ does not passes through small square $j$. 

Separately summing $Y_{i}^{j}$ and $H_{i}$ over the indices $i$ and $j$, we have
\begin{equation}
    \sum_{i=1}^{N_{q}(n)}\sum_{j=1}^{m}Y_{i}^{j}=\sum_{i=1}^{N_{q}(n)}H_{i},
    \label{P23}
\end{equation}
where $N_{q}(n)=\frac{n(1-q)}{2}$.

Taking the expectation of both sides of (\ref{P23}), we obtain
\begin{equation}
m N_{q}(n) E\left[ Y_{i}^{j}\right]=N_{q}(n)E\left[ H_{i}\right]\iff m E\left[ Y_{i}^{j}\right]=E\left[ H_{i}\right].
    \label{P24}
\end{equation}

Combining (\ref{P21}) and (\ref{P24}), we have
\begin{equation}
    \text{Pr}\{Y_{i}^{j}=1\}=\Theta(r_{q}(n)).
    \label{P25}
\end{equation}

According to  (\ref{P25}),  the  number of S-D lines passing through the small square $j$  is 
\begin{equation}
    Y^{j}=\sum_{i=1}^{N_{q}(n)}Y_{i}^{j}.
    \label{P26}
\end{equation}

From  (\ref{P26}),  the expectation of $Y^{j}$ is given by
\begin{equation}
    E[Y^{j}]=\Theta\left(N_{q}(n)r_{q}(n)\right)=\Theta\left(\sqrt{n(1-q)\log n}\right).
    \label{P27}
\end{equation}

By applying Chebyshev's inequality, we get
\begin{equation}
    \text{Pr}\{Y^{j}>(1+\sigma)E[Y^{j}]\}\leq e^{-\frac{E[Y^{j}]\sigma^{2}}{4}}.
    \label{P28}
\end{equation}

Substituting (\ref{P27}) into (\ref{P28}), we have
\begin{equation}
    \text{Pr}\{Y^{j}>(1+\sigma)E[Y^{j}]\}\leq \frac{1}{n^{2}},
    \label{P29}
\end{equation}

where $\sigma=2\sqrt{\frac{2\log n}{E[Y^{j}]}}$.

Since $\lim\limits_{n\rightarrow +\infty}\frac{1}{n^{2}}=0$, we get

\begin{equation}
    Y^{j}\leq (1+\sigma)E[Y^{j}]=O\left(\sqrt{(1-q) n\log n}\right),\: j=1,\cdots,m.
    \label{P30}
\end{equation}

Using the inequality $\text{Pr}\{\cup_{i=1}^{m} A_{i}\}\leq \sum_{i=1}^{m}\text{Pr}\{A_{i}\}$, the probability that any small square serves more than $(1+\sigma)E[Y^{j}]$ S-D lines converges to 0, i.e., the  number of S-D lines served by each small square is $O\left(\sqrt{(1-q) n\log n}\right)$.

Using (\ref{P30}), the feasible data transmission rate is given by

\begin{equation}
    \lambda_{q}(n)\geq \frac{c}{\sqrt{(1-q)n\log n}},
\end{equation}
where $c=\Theta \left(\frac{W}{M^{2}}\right)$. 

According to Definition \ref{D3} and $N_{q} (n)=\frac{n(1-q)}{2}$,   the lower bound of  network capacity is 

\begin{equation}
    S_{q}(n)=\Omega\left(\sqrt{\frac{n(1-q)}{\log n}}\right).
    \label{L1}
\end{equation}

\ding{193} Upper bound

The upper bound on  network capacity asserts that no interference-free scheduling and routing scheme can achieve a capacity beyond this bound. According to the protocol model in Section \ref{s221},  the area consumed for a single-hop successful transmission is at least $\frac{\pi\left(\Delta^{\prime}r_{q}(n)\right)^{2}}{4}$. Thus, the total number of transmitter-receiver pairs that can successfully transmit  in the same time slot,  denoted by $F^{u}(n)$, is at most

\begin{equation}
   F^{u}(n) =\frac{4}{\pi\left(\Delta^{\prime}r_{q}(n)\right)^{2}}.
    \label{U1}    
\end{equation}

Let  $\overline{d}$  be the average Euclidean distance between source node and destination node. Since nodes are uniformly distributed in the unit square, each data packet requires at least $\frac{\overline{d}-o(1)}{r_{q}(n)}$ hops to reach the destination node. Therefore, the number of bits that the network needs to successfully transmit is at least
\begin{equation}
   F(n)\geq n(1-q)\lambda_{q} (n)\frac{\overline{d}-o(1)}{r_{q}(n)}.
    \label{U2}    
\end{equation}

Each node successfully transmit $W$ bits per time slot. Combining (\ref{U1}) and (\ref{U2}), we have
\begin{equation}
\begin{aligned}
F(n) &\leq W F^{u}(n) \\
&\iff \lambda_{q}(n)\frac{n(1-q)(\overline{d}-o(1))}{r_{q}(n)} \\
&\leq \frac{4W}{\pi\left(\Delta' r_{q}(n)\right)^{2}}.
\end{aligned}
\label{U3}
\end{equation}
where $\lambda_{q}(n)$ is the feasible data transmission rate.

Substituting $r_{q}(n)= \sqrt{\frac{\log n+\xi}{\pi (1-q)n}}$ from Lemma \ref{LL1} into  (\ref{U3}), we get
\begin{equation}
    \lambda_{q}(n)\leq \frac{c_{1}}{\sqrt{(1-q)n\log n}},
    \label{U4}
\end{equation}
where $c_{1}=\frac{4W}{\pi\Delta^{\prime}(\overline{d}-o(1))}$. 

Combining Definition \ref{D3}, the upper bound of  network capacity is 
\begin{equation}
    S_{q}(n)=O\left(\sqrt{\frac{n(1-q)}{\log n}}\right).
    \label{U5}
\end{equation}

From the above proof, for any network scheduling strategy and network  topology, (\ref{U5})  holds.

Combining (\ref{L1}) and (\ref{U5}), the network capacity is given by 
\begin{equation} 
S_{q}(n)=\Theta\left( \sqrt{\frac{n(1-q)}{\log n}}\right)\:\:\text{bits/s}.
\label{TH1251013}
\end{equation}

As \( q \to 0 \), Equation~(\ref{TH1251013}) reduces to $\Theta\left(\sqrt{\frac{n}{\log {n}}}\right)\:\:\text{bits/s}$, which is the classical result  in \cite{ref13}.

\end{proof}
We have the following important insights  from Theorem \ref{Th1}.

\begin{itemize}
\itemindent 2.8em
\item  The degradation of network capacity due to node failures is proportional to $\Theta\left(\sqrt{1-q}\right)$.   This is  because the minimum transmission radius required to maintain the connectivity of the  network  topology increases, thereby reducing the total number of active transmitter-receiver pairs  in the same  time slot and the number of hops from the source node to the destination node. 
\item The feasible data transmission rate  is $$\lambda_{q}(n)=\Theta\left(1/\sqrt{(1-q)\log n}\right)\:\:\text{bits/s}.$$ When $n$ is fixed, the feasible data transmission rate increases with the node failure probability. This is because the interference within the network decreases, thereby reducing the intensity of competition for communication resources. Moreover, $\lambda_{q}(n)$ converges to zero as the number of nodes increases, indicating that the feasible data transmission rate is not scalable under node failures. 
\item Increasing the number of redundancy  nodes to improve network capacity is not feasible. Let $n_{1}$ redundancy  nodes be deployed in the network and are responsible only for relaying data packets from other nodes. Assume all redundancy  nodes do not fail, i.e., they are available with probability 1 for constructing the network  topology. From Theorem \ref{Th1}, the network capacity after adding redundancy  nodes  is 
\begin{equation}
   S_{q}(n+n_{1})=\Theta\left(\sqrt{\frac{n(1-q)+n_{1}}{\log (n+n_{1})}}\right)\:\:\text{bits/s}.
   \label{eq251013}
\end{equation}

\:\:\:\:\:The feasible data transmission rate after adding redundancy  nodes is \scalebox{0.5}{$\Theta\left(\frac{1}{n(1-q)}\sqrt{\frac{np+n_{1}}{\log\: (n+n_{1})}}\right)$} bits/s. A very large number of redundancy nodes is required to increase the network capacity to $\omega+1$ times its original value. We provide an numerical example  as shown in Fig. \ref{Fig6R}.   For example,     when $q=0.2$, it requires  6657 redundancy  nodes to achieve three times the network capacity of a wireless ad hoc network with 1000 nodes. This number is much  greater than the initial node scale. In other words, it can be seen that adding redundancy  nodes enhances the connectivity of the network  topology rather than improving the network capacity. 
\begin{figure}[!t]
\centering

% 第一行
\begin{minipage}[t]{0.48\linewidth}
    \centering
    \includegraphics[width=\linewidth]{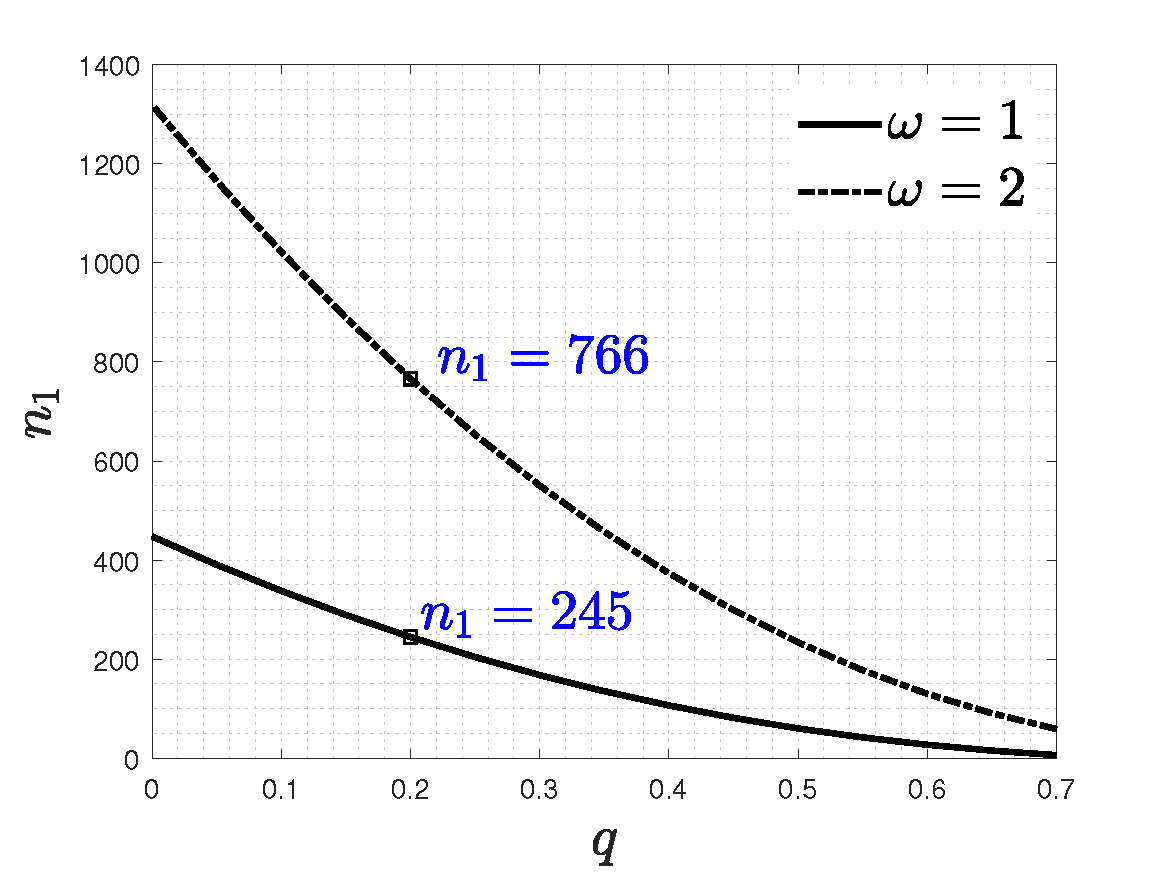}\\
    (a) $n=100$
\end{minipage}
\hfill
\begin{minipage}[t]{0.48\linewidth}
    \centering
    \includegraphics[width=\linewidth]{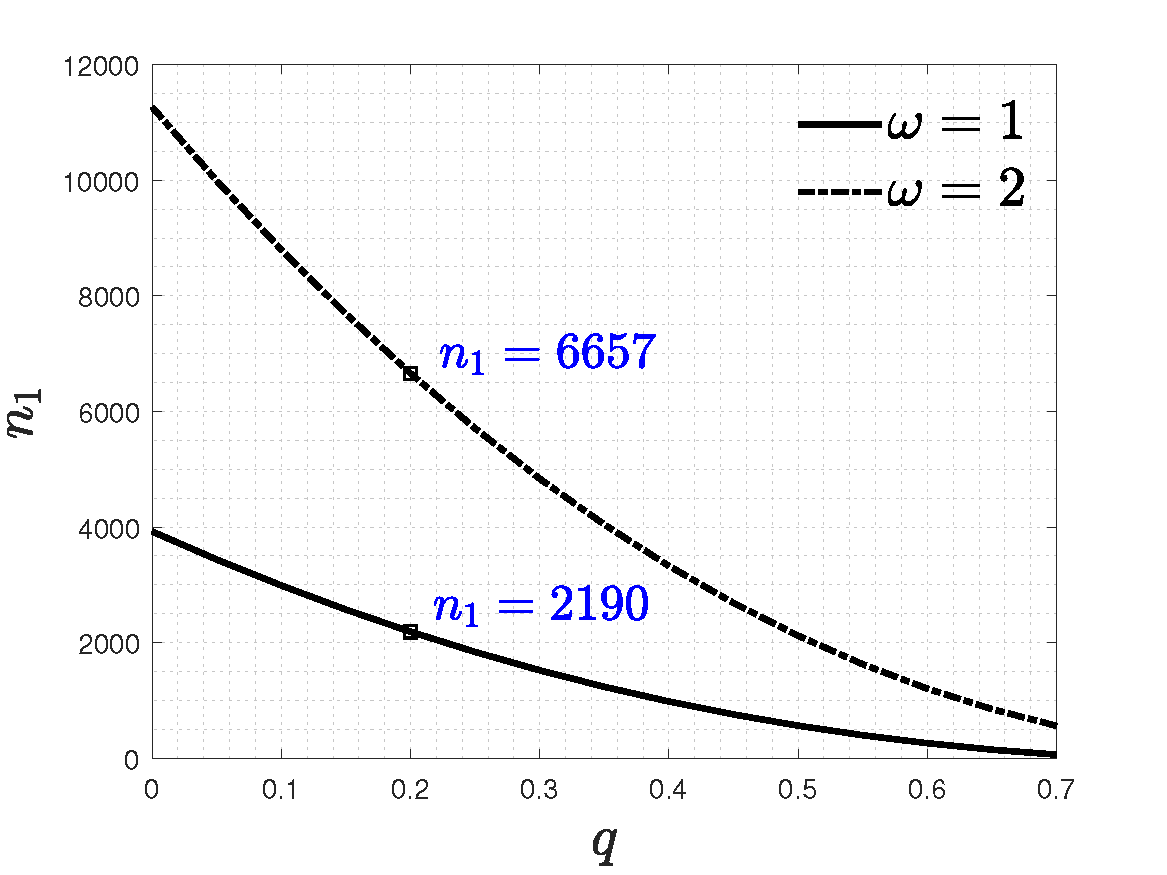}\\
    (b) $n=1000$
\end{minipage}

\vspace{0.3cm}

% 第二行
\begin{minipage}[t]{0.48\linewidth}
    \centering
    \includegraphics[width=\linewidth]{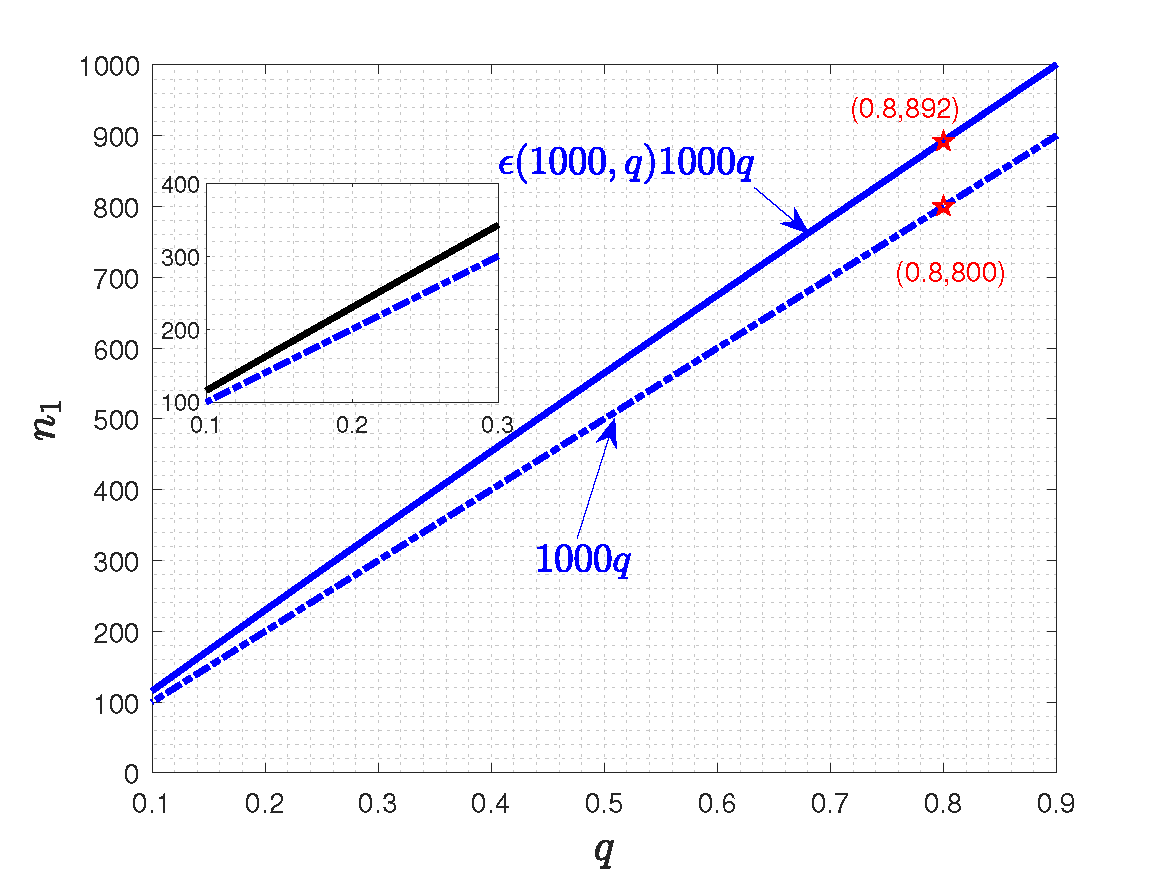}\\
    (c) $n_{1}$ vs.\ $q$
\end{minipage}
\hfill
\begin{minipage}[t]{0.48\linewidth}
    \centering
    \includegraphics[width=\linewidth]{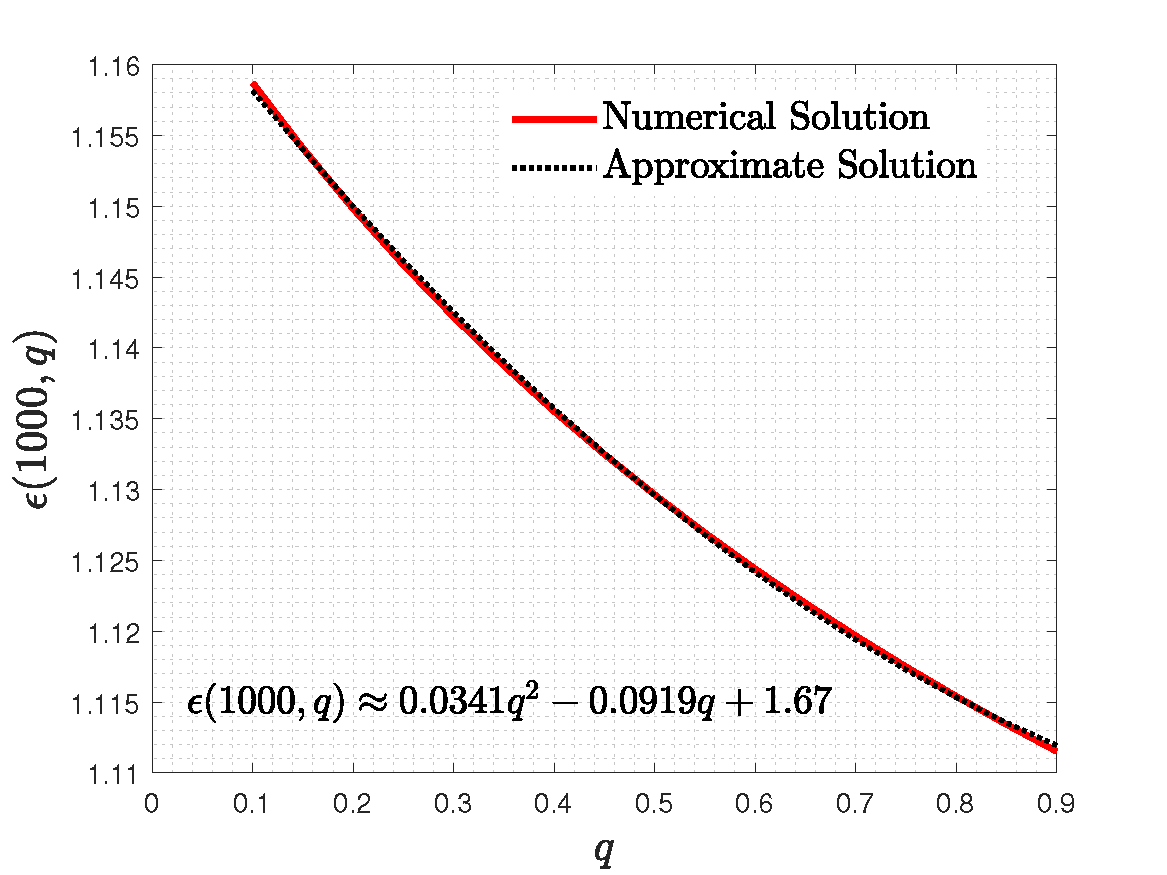}\\
    (d) $\epsilon(1000,q)$ vs.\ $q$
\end{minipage}

\caption{The number of redundancy nodes varies with the node failure probability.}
\label{Fig6R}
\end{figure}

\end{itemize}

\subsection{Some   corollaries}
Based on the previous analysis, some important corollaries are given as follows.

\begin{corollary}To maintain the network capacity at \scalebox{0.65}{$\Theta\left(\sqrt{\frac{n}{\log {n}}}\right)$ \:bits/s} in  wireless ad hoc networks with node failure probability,  at least $n_{1} = \epsilon(n,q) nq$ redundant nodes should be deployed, where $\epsilon(n,q)>1$ is the ratio  of the number of redundant nodes to $nq$. \label{C2510}
\end{corollary}
\begin{proof}  Combining with (\ref{eq251013}), we have
$$S(n)= S_{q}(n+n_{1})\Longrightarrow  \frac{n}{\log {n}}=\frac{n(1-q)+n_{1}}{\log (n+n_{1})}\:\:\text{bits/s},$$
where $n_{1}=\epsilon(n,q) nq$. $\epsilon(n,q)>1$ can be verified by a simple numerical computation. 

\end{proof}

For $n=1000$,  the  number of redundant nodes under varying node failure probabilities is greater than $1000q$, as demonstrated in Fig.  \ref{Fig6R}(c).   $\epsilon(n,q)$ is a monotonically decreasing function of the node failure probability---see Fig. \ref{Fig6R}(d). Given the node failure probability $q>0$, the number of redundant nodes decreases non-linearly as the number of nodes ($n$). The relationship between $\epsilon(n,0.2)$ and $n$ under  $q=0.2$ is shown in Fig. \ref{FigRV0406}a. When  $n$ is sufficiently large or $q$ is sufficiently small, deploying approximately $nq$ redundant nodes is sufficient to maintain the network capacity at \scalebox{0.65}{$\Theta\left(\sqrt{\frac{n}{\log n}}\right)$ \:bits/s}.

\begin{figure}[!t]
\centering

\begin{minipage}[t]{0.85\linewidth}
    \centering
    \includegraphics[width=\linewidth]{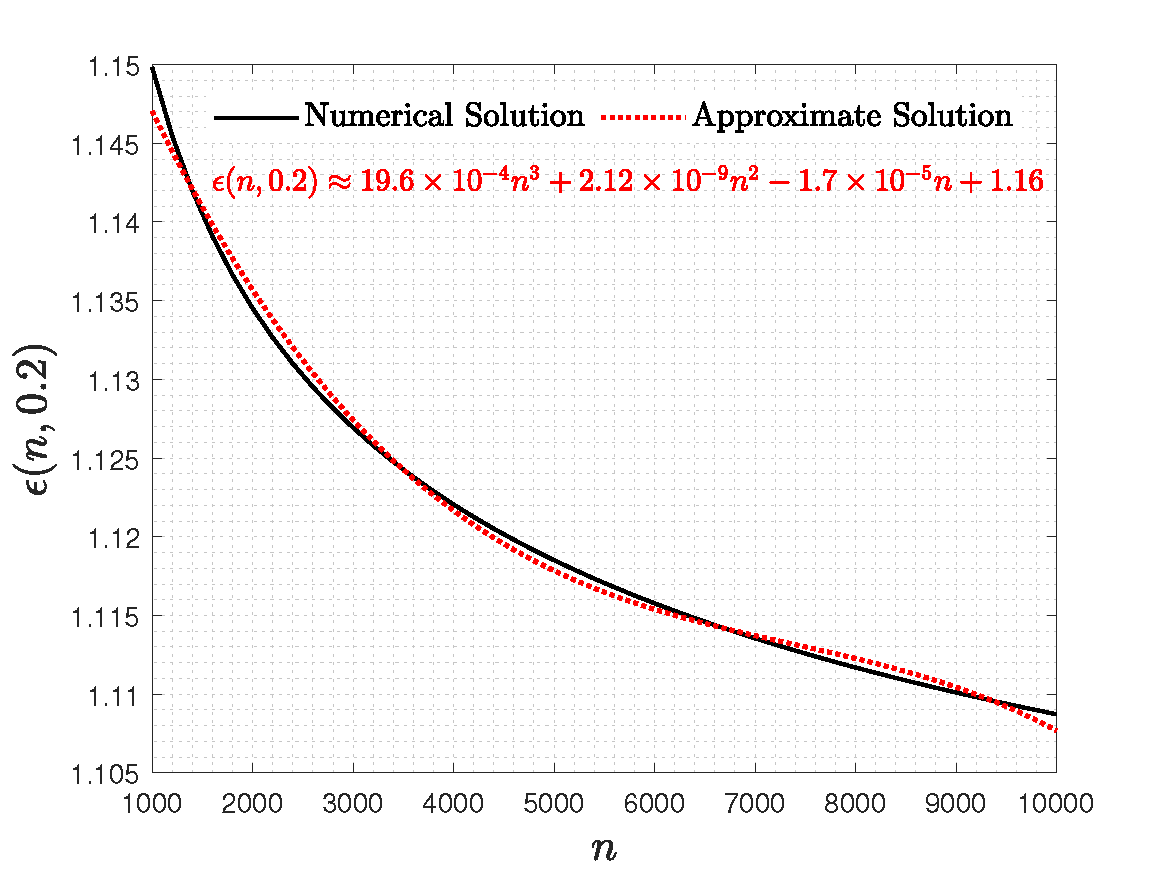}\\
    (a) $\epsilon(n,0.2)$ vs.\ $n$
\end{minipage}
\hfill
\begin{minipage}[t]{0.85\linewidth}
    \centering
    \includegraphics[width=\linewidth]{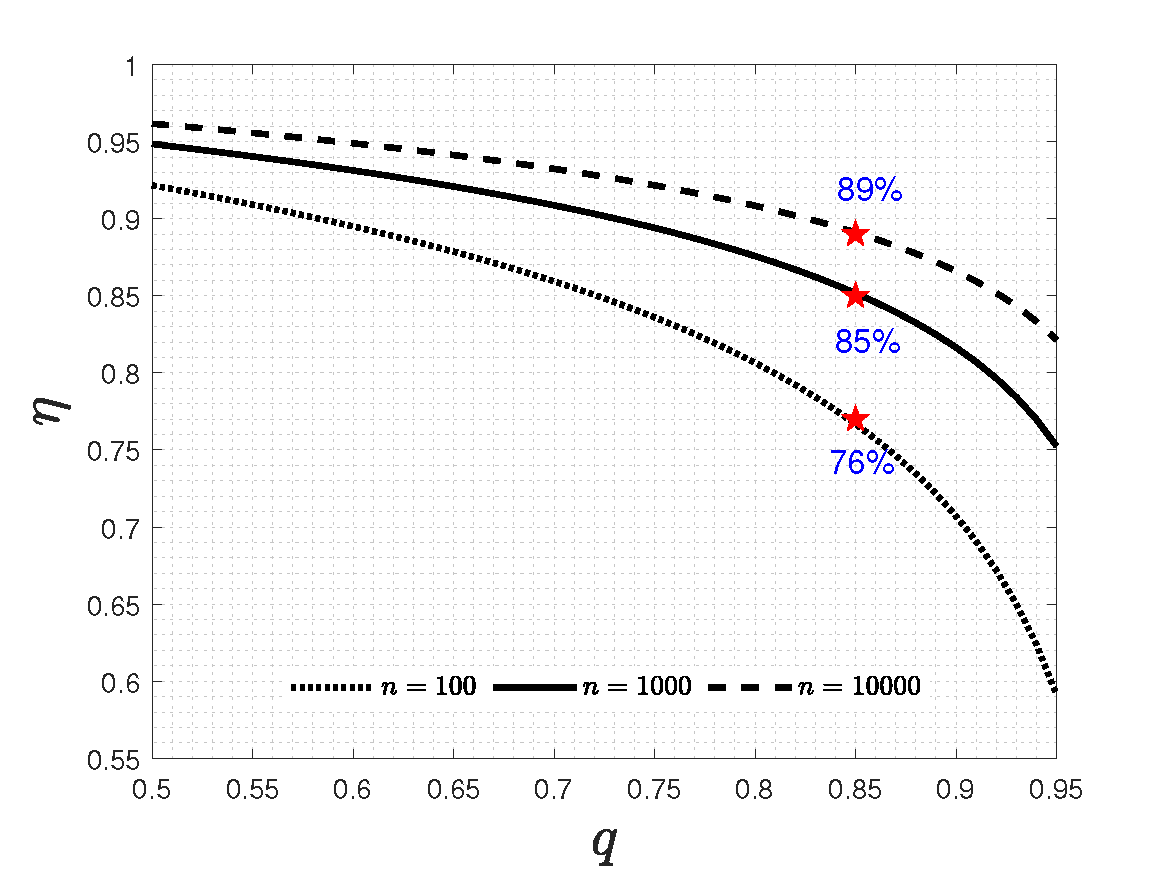}\\
    (b) 
\end{minipage}

\caption{Network capacity loss vs.\ node failure probability.}
\label{FigRV0406}
\end{figure}

\begin{corollary}The wireless ad hoc networks with $n$ nodes and node failure probability $q$ have lower network capacity than one with $n(1-q)$ nodes without considering node failures, i.e.,
\begin{equation*}
\eta \triangleq \sqrt{\frac{\log n(1-q)}{\log n}}\leq 1.
\label{C1}
\end{equation*}
\label{Corollary1}
\end{corollary}
\begin{proof}  If node failures is not considered, the network capacity corresponding to deploying  $n(1-q)$ nodes is \scalebox{0.65}{$S(n(1-q) )=\Theta\left(\sqrt{\frac{n(1-q)}{\log n(1-q)}}\right )\:\:\text{bits/s}$}, which can be derived using the proof method from \cite{ref13}. From Theorem 1, the network capacity of  wireless ad hoc networks with node size $n$ and the node failure probability $q$ is \scalebox{0.65}{$S_{q} (n) =\Theta\left(\sqrt{\frac{n(1-q)}{\log n}}\right)\:\:\text{bits/s}$}. In this case, the expected number of non-faulty nodes is $n(1-q)$. Based on the above results, Corollary \ref{Corollary1} is given.

\end{proof}

$\eta\leq 1$  holds because the critical transmission radius is given by \scalebox{0.65}{$r(n(1-q))=\sqrt{\frac{\log n(1-q)}{n(1-q)}}$}  when node failures are not considered (see Lemma \ref{LL1}). When node failures are considered,  however,  the critical transmission radius is \scalebox{0.65}{$r_{q}(n)=\sqrt{\frac{\log n}{n(1-q)}}$}. By combining the protocol model from Section \ref{s221},  the number of simultaneous transmit-receive node pairs supported by the network in our work is less than that reported in \cite{ref13}. This means that the loss of network capacity is the cost incurred to overcome the randomness of node failures. Under the same node failure probability, a greater number of nodes results in reduced the loss of network capacity, as shown in Fig. \ref{FigRV0406}b. When $q = 0.85$, for example, the network capacity loss rates corresponding to node scales of 100, 1000, and 10,000 are 24\%, 15\%, and 11\% respectively, compared to scenarios where node failures are not considered.

Using the same proof technique as Theorem \ref{Th1}, we can  derive the network capacity of wireless ad hoc networks in the three-dimensional space.

\begin{corollary}The network capacity of wireless ad hoc networks  deployed in three-dimensional space is 
$$S_{q}(n)=\Theta\left( \sqrt[\frac{3}{2}]{\frac{n(1-q)}{\log n}}\right)\:\:\text{bits/s},$$ 
where $q$ is the node failure probability. When $q=0$, our result degenerates to the classical result established in  \cite{ref181}.
\label{C22}
\end{corollary}

With the same node failure probability, from Corollary \ref{C22} and Theorem \ref{Th1},  wireless ad hoc networks in three-dimensional space achieve higher network capacity than those in two-dimensional planes.

\subsection{Delay}
\begin{theorem}The delay of wireless ad hoc networks  is given by
$$D_{q}(n)=\Theta\left( \sqrt{\frac{n(1-q)}{\log n}}\right),$$
where $q$ is the node failure probability. 
\label{Th2}
\end{theorem}

\begin{proof} According to the law of large numbers, the average distance between source node and destination node is  $$\overline{d}=\frac{\sum_{i=1}^{n(1-q)}d(i)}{n(1-q)}=\Theta(1),$$ where $d_{i}$ is the distance between source-destination pair $i$. 

Since each hop distance is $\Theta\left(r_{q}(n)\right)$, the average number of hops is \scalebox{0.5}{$\Theta\left(\frac{1}{r_{q}(n)}\right)$}. The delay is proportional to the average number of hops because the delay required for each hop is constant and independent of the number of nodes. Summarizing the above analysis,  the delay  is given by
$$D_{q}(n)=\Theta\left( \sqrt{\frac{n(1-q)}{\log n}}\right),$$
where $q$ is the node failure probability. 

\end{proof}

 To reduce the delay, reducing the number of hops is a feasible strategy by increasing the transmission radius. Meanwhile, the number of transmitter-receiver pairs that can successfully transmit  within the same time slot can also decrease (see Fig. \ref{Fig1}). This can degrade the network capacity. Thus, there exists an inherent  trade-off between delay and network capacity.

\section{Optimal trade-off between network capacity and delay}

Combining Theorems \ref{Th1} and \ref{Th2} in Section \ref{Set3}, the ratio of network capacity to delay is given by

\begin{equation}
\frac{S_{q}(n)}{D_{q}(n)}=O\left(1\right).
\label{S31}
\end{equation}

In terms of order, is the trade-off relationship given by Eq. (\ref{S31}) optimal?   To answer this fundamental question, the definition of the optimal trade-off between network capacity and delay is given as follows \cite{ref19}.

\begin{definition} Let the network capacity and delay of wireless networks be denoted by $S_{q}(n)$ and $D_{q} (n)$, respectively. The ratio $\frac{S_{q}(n)}{D_{q}(n)}$ is optimal if the following conditions are satisfied.

a) There exists a multi-hop  transmission strategy $\mathcal{M}$ such that the  network capacity and delay satisfy $S_{q}^{\mathcal{M}}(n)=\Theta\left(S_{q}(n)\right)$ and $D_{q}^{\mathcal{M}}(n)=\Theta\left(D_{q}(n)\right)$.

b) For any multi-hop  transmission strategy $\mathcal{M^{\prime}}$, $S_{q}^{\mathcal{M^{\prime}}}(n)=\Omega\left(S_{q}(n)\right)$ and $D_{q}^{\mathcal{M^{\prime}}}(n)=\Omega \left(D_{q}(n)\right)$,
where $q$ is the node failure probability.
\label{D6}
\end{definition} 

In Definition \ref{D6}, condition (a) ensures that there exists some strategy capable of achieving the corresponding network capacity and delay, while condition (b) indicates that for all strategies with multi-hop  transmission properties, the network capacity and delay are $S_{q} (n)$ and $D_{q} (n)$, respectively.

\begin{theorem} For all multi-hop transmission strategies, the optimal trade-off between network capacity and delay is given by
$$\frac{S_{q}(n)}{D_{q}(n)}=O\left(1\right),$$
where $q$ is the node failure probability.
\label{Th3}
\end{theorem}

\begin{proof} Let the average distance of S-D lines be $\overline{d}$ and the feasible data transmission rate be $\lambda_{q}(n)$. After $T$ time slots,  the total number of bits carried by the wireless network is $e=\lambda_{q}(n)N_{q}(n)T$. The total number of hops required for bit $b$ from source to destination be $h(b)$ and  the $h$-th
hop distance  for bit  $b$ be $r_{b} (h)$. Based on $\sum_{h=1}^{h(b)}r_{b} (h)\geq \overline{d}$, we have
\begin{equation}
    \sum_{b=1}^{e}\sum_{h=1}^{h(b)}r_{b} (h)\geq e\overline{d}.
    \label{T32}
\end{equation}

Combining the protocol model from Section \ref{s221}, the area consumed for the successful transmission of each bit $b$ in the $h$-th hop is $\frac{\pi(\Delta^{\prime}r_{b}(h))^{2}}{4}$. The number of bits  transmitted from a source node within $T$ time slots is at most $WT$, where $W$  is the number of bits transmitted per  time slot by each source node. Hence, we have
\begin{equation}
\sum_{b=1}^{e}\sum_{h=1}^{h(b)}\frac{\pi(\Delta^{\prime}r_{b}(h))^{2}}{4}\leq WT\times 1.
    \label{T33}
\end{equation}

Let the total number of hops taken for the successful transmission of all bits be $H=\sum_{b=1}^{e}h(b)$. Since  $f(x)=x^{2},\:\: x>0$ is a convex function, we obtain
\begin{equation}
\left(\sum_{b=1}^{e}\sum_{h=1}^{h(b)} \frac{1}{H}r_{b}(h)\right)^{2}\leq \sum_{b=1}^{e}\sum_{h=1}^{h(b)} \frac{1}{H}\left(r_{b}(h)\right)^{2}.
    \label{T34}
\end{equation}

Combining (\ref{T32})-(\ref{T34}), the inequality is given by

\begin{equation}
    \left(\lambda_{q}(n)N_{q}(n)T\overline{d}\right)^{2}\leq \frac{4WT}{\pi(\Delta^{\prime})^{2} }H.
    \label{T35}
\end{equation}

From the above proof, it can be seen that for all multi-hop  transmission strategies, Eq. (\ref{T35}) holds with probability 1. Furthermore,  the average number of hops taken for each bit to successfully transmit from the source node to the destination node be $\overline{h}=\frac{H}{e}$. Substituting $\overline{h}$ into (\ref{T35}), we have

\begin{equation}
\lambda_{q}(n)N_{q}(n)\leq \frac{4W}{\pi (\Delta^{\prime}\overline{d})^{2}} \overline{h}.
    \label{T36}
\end{equation}

 Taking the expectation of both sides of (\ref{T36}), we obtain

\begin{equation}
    \lambda_{q}(n)E\left[N_{q}(n)\right]\leq E\left[\frac{4W}{\pi (\Delta^{\prime}\overline{d})^{2}} \overline{h}\right]=c\overline{h},
    \label{T37}
\end{equation}
where $c=\frac{4W}{\pi (\Delta^{\prime}\overline{d})^{2}}$,  which is independent of $n$.

By the condition (a) of Definition \ref{D6},  there exists a multi-hop  transmission strategy $M$ such that the following equation holds, i.e.,
\begin{equation}
    S_{q}(n)=\Theta \left( E[N_{q}(n)]\lambda_{q}(n)\right).
    \label{T38}
\end{equation}

Using (\ref{T37}) and (\ref{T38}), we have
\begin{equation}
    S_{q}(n)\leq c\overline{h} \iff \frac{S_{q}(n)}{ \overline{h}}=O(1).
    \label{T39}
\end{equation}

Since the delay for each hop is constant, the delay is proportional to the average number of hops, i.e., $ D_{q}(n)=\Theta \left(\overline{h}\right)$. Combining the above results, the optimal trade-off between network capacity and delay is 
$$\frac{S_{q}(n)}{D_{q}(n)}=O\left(1\right).$$
\end{proof}

In the sense of  order,   Theorem \ref{Th3}  holds for any multi-hop  transmission strategy. This indicates that increasing the network capacity comes at the cost of increasing the delay  and is independent of the node failure probability and spatial dimension.   For the scenarios with lower delay, using one-hop or fewer transmissions may be a preferable strategy. For the scenarios with high traffic load, multi-hop  transmission represents the optimal strategy. Using the multi-hop  transmission strategy cannot  simultaneously achieve both lower delay and higher network capacity.  

According to  Theorem \ref{Th3}, increasing the network capacity by adding redundancy relay nodes can also increase the delay. This further indicates that the primary role of adding redundancy relay nodes is to enhance the connectivity of the network  topology. Additionally, the conclusion of Theorem 3 holds regardless of the spatial dimension.

\begin{remark}
Compared with \cite{ref14}–\cite{ref24R}, Theorems \ref{Th1}, \ref{Th2} and \ref{Th3} respectively  establish the network capacity, delay and their optimal trade-off in wireless ad hoc networks subject to random node failures after deployment. Specifically, they quantitatively characterize the impact of random node failures on both network capacity and delay.  We also prove that the optimal trade-off between network capacity and delay is $\Theta(1)$, which is independent of the node failure probability.   
\end{remark}

\section{Conclusions}
In this paper, the network capacity, delay and their optimal trade-off relationship are provided when wireless ad hoc networks   encounter node failures.  We find that at least $\epsilon(n,q)nq$ redundant nodes need to be added to maintain the network capacity at \scalebox{0.65}{$\Theta\left(\sqrt{\frac{n}{\log n}}\right)$}\: bits/s, where $\epsilon(n,q)>1$. It is also shown that  any strategy with multi-hop properties cannot simultaneously achieve high network capacity and low delay. These results provide a preliminary theoretical explanation of how node failures affect the fundamental performance of large-scale wireless ad hoc networks.  In future work, we will design networking protocols by considering specific requirements, including network status, application scenarios and task requirements.

\appendix

\subsection{Network connectivity  under  realistic wireless channel conditions \label{AppendixA1}}

In the sensitivity analysis of network connectivity, the key wireless channel parameters considered are as follows:
\begin{itemize}
  \item \textit{Path loss:}  
 Large-scale attenuation of signal power as the distance between transmitter and receiver increases. The path loss function is typically modeled as  
  \(
  \ell(x) = x^{-\alpha},
  \)
  where \(\alpha\) is the path loss exponent.

  \item \textit{Log-Normal shadowing:}  
  Slow variations in signal strength caused by obstacles, resulting in fluctuations around the average path loss. The lognormal spread $\sigma$ follows a log-normal distribution $\mathcal{N}(0, \sigma^2)$.

  \item \textit{Rayleigh fading:}  
  Rapid small-scale fluctuations in signal amplitude due to multipath scattering in non-line-of-sight environments. The power gain \(G \) follows an exponential distribution,  
  \(
  G \sim \text{Exponential}(1),
  \)
\end{itemize}

Based on the above parameters, the received power at a distance \( d \) from a transmitting node is given by
\begin{equation}
   P_r(d) = P_t - 10 \alpha \log_{10}(d) + \sigma+G, 
\end{equation}
where \( P_t \)  is the transmission power and \( P_{\min} \) is the minimum required received power for successful communication. 

A communication link can be successfully established if \(  P_r(d) > P_{\min}  \); otherwise, it cannot.
 When \( \sigma = 0,G=0 \) and \( \alpha = 2 \), the connection rule reduces to a geometric condition (see~(\ref{Connec0121})).  By integrating key parameters of the wireless channel, we perform simulations to evaluate the network connectivity under different parameter settings (see Fig. \ref{Fig8R}).  Simulation results indicate that the set of links successfully  established under the geometric rule is a subset of those formed under more realistic wireless channel conditions, including path loss, shadowing, and Rayleigh fading. This confirms that the analysis based on Lemma \ref{LL1} is robust to the random fluctuations of wireless channel. As a result, the corresponding capacity results remain valid under practical wireless environments. As shown in Fig. \ref{Fig8R} (c) and \ref{Fig8R} (d), a higher path loss exponent corresponds to lower network connectivity. Moreover, the randomness of the wireless channel does not affect the asymptotic scaling of the critical transmission radius as stated in Lemma \ref{LL1}. 
\begin{figure}[!ht]
\centering
\subfloat[$\sigma=5$ dBm]{\includegraphics[width=0.5\linewidth]{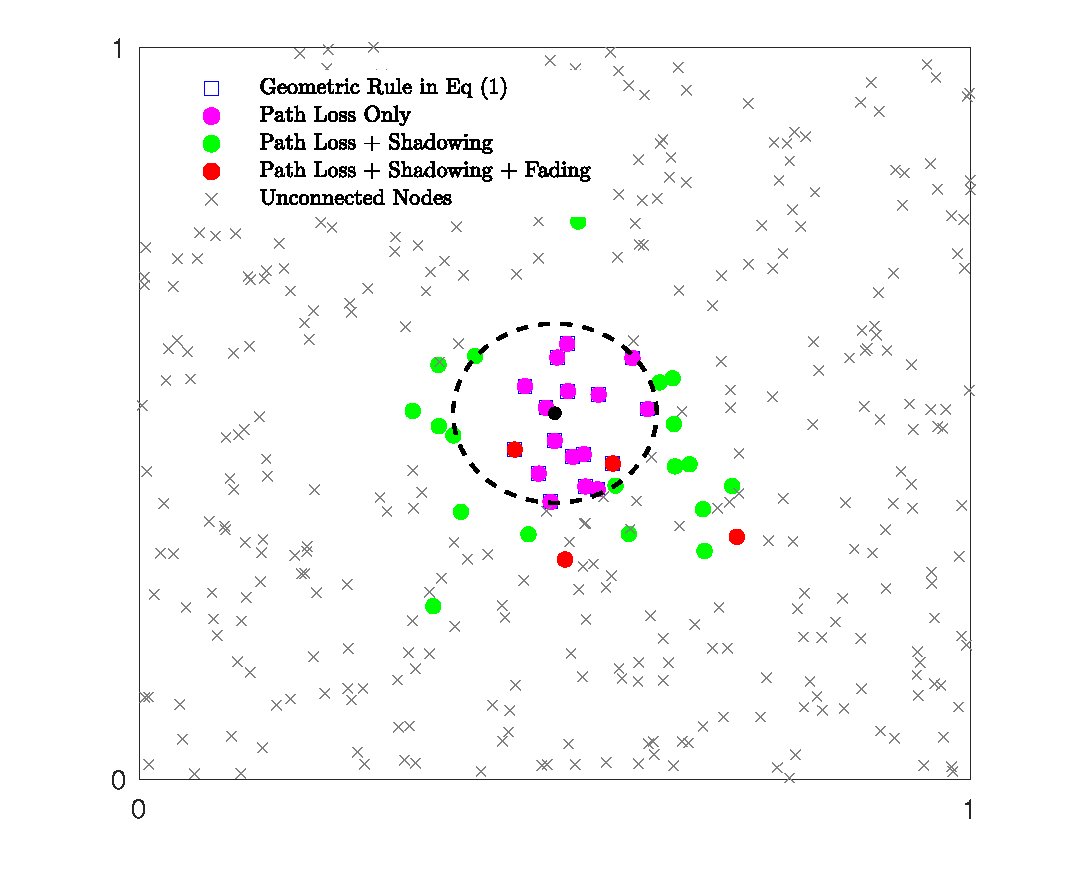}}
\hfill
\subfloat[]{\includegraphics[width=0.5\linewidth]{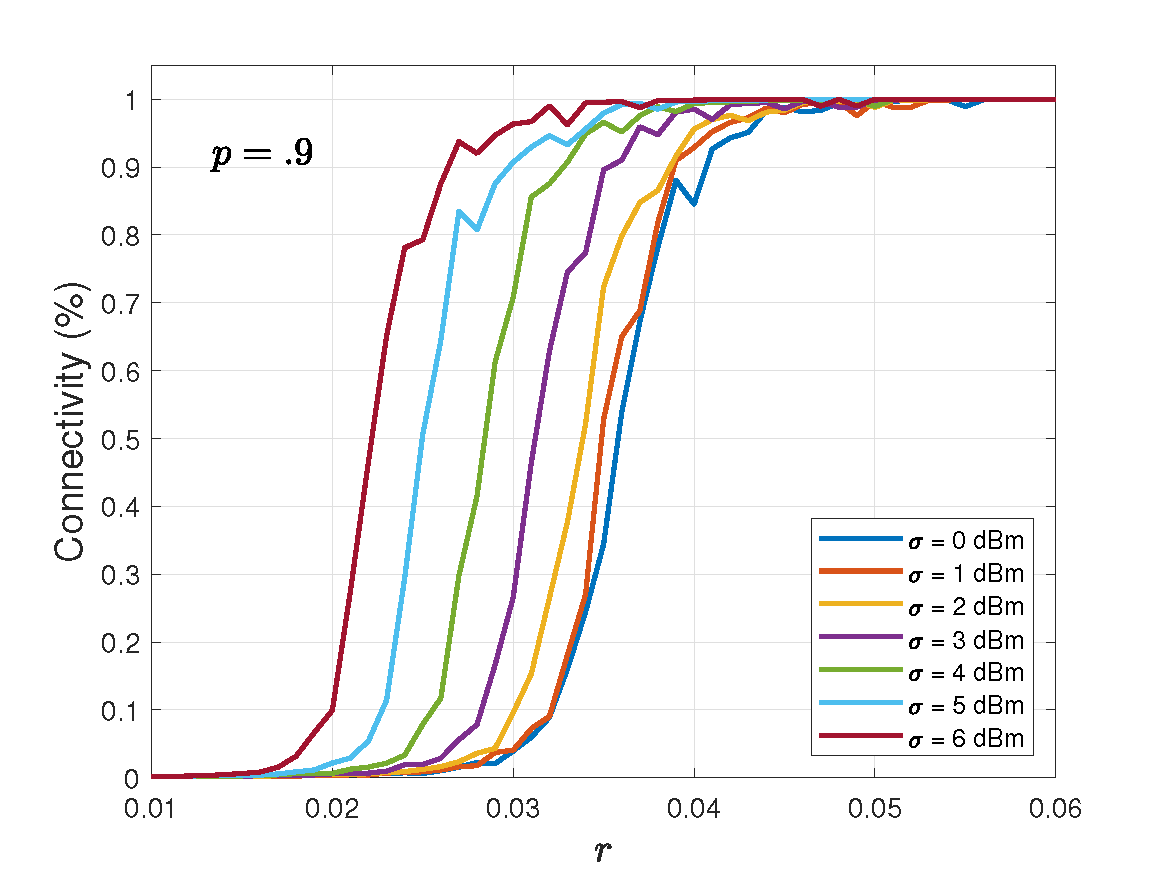}}

\subfloat[]{\includegraphics[width=0.5\linewidth]{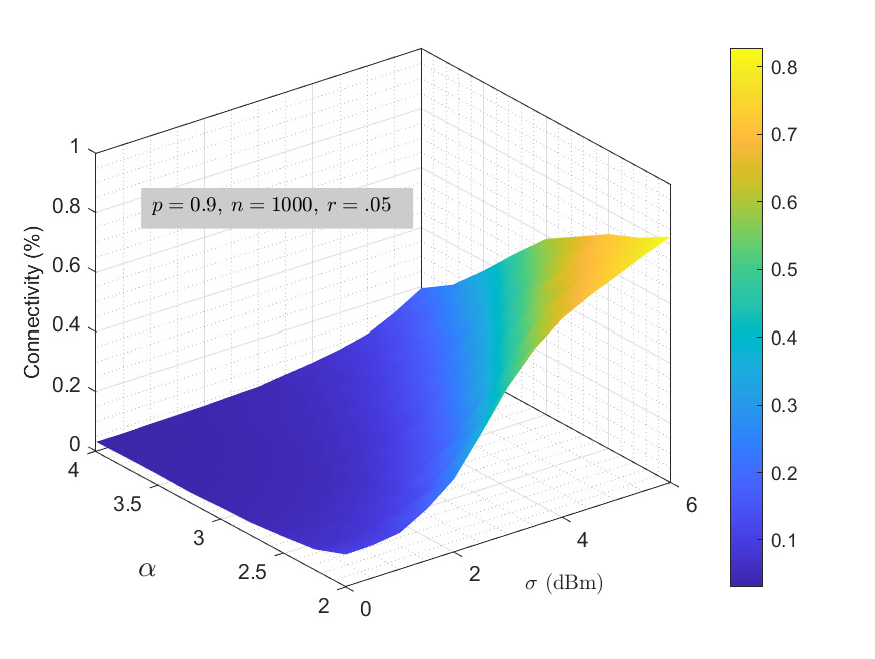}}
\hfill
\subfloat[]{\includegraphics[width=0.5\linewidth]{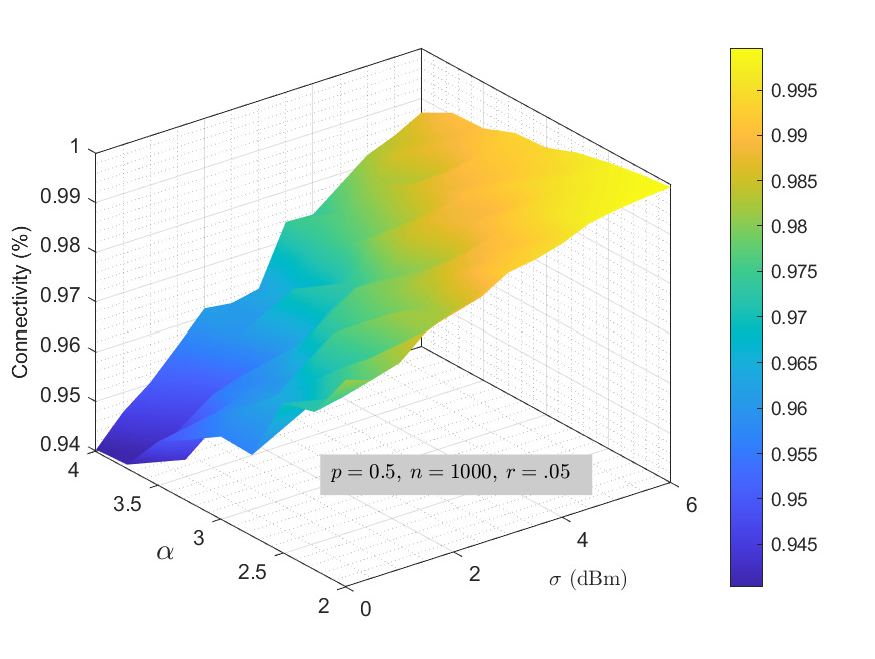}}

\caption{Network connectivity under different combinations of path loss exponent ($\alpha$) and lognormal spread ($\sigma$). $P_{min}=-80\:\text{dBm}$, $n=1000$, $\alpha=2.5$}
\label{Fig8R}
\end{figure}

\subsection{Proof of proposition \ref{P1} \label{AppendixA}}
Let the  number of non-faulty nodes in the small square $i$   be $Z_{n_{0}}$. Since each node independently participates in the construction of the network  topology with probability $p=1-q$, the  number of non-faulty nodes is $n_{0}=n(1-q)$. Let $I_{j}$  be the indicator  for whether node $j$ falls within the small square $i$,  which is defined by
\begin{equation}
I_j = 
\begin{cases} 
0, & \text{node } j \text{ falls within small square } i, \\\\
1, & \text{otherwise}.
\end{cases}, 
\label{A1212}
\end{equation}
where $\quad j = 1, \dots, n_0$.

According to the geometric probability model,  $\text{Pr}\{I_{j}=1\}=a_{q}^{2}(n)$ and $\text{Pr}\{I_{j}=0\}=1-a_{q}^{2}(n)$.   Using  (\ref{A1212}), $Z_{n_{0}}$ is given by   
\begin{equation}        
Z_{n_{0}}=\sum_{j=1}^{n_{0}}I_{j}.		 
\label{eq9}  
\end{equation}

From the above analysis, $Z_{n_{0}}$ follows a binomial distribution with parameters ($p_{n},n_{0}$), where $p_{n}=\text{Pr}\{I_{j}=1\}$. Using Chebyshev's inequality, we obtain
\begin{equation}    
\text{Pr}\{Z_{n_{0}}\geq (1+\sigma)c_{0}\log n\}\leq e^{-\frac{c_{0}\sigma^{2}}{1+\sigma}\log n}=n^{-\frac{c_{0}\sigma^{2}}{1+\sigma}},    
\label{eq10}
\end{equation}
where $c_{0}=\frac{1}{\pi (1-q)}$ and $0<\sigma<1$.

Let  $Q$ denote as the event that at least one small square contains more than $(1+\sigma) c_{0} \log n$ non-faulty nodes. Using $\text{Pr}\{\cup_{i=1}^{m} A_{i}\}\leq \sum_{i=1}^{m}\text{Pr}\{A_{i}\}$ and (\ref{eq10}), we have
\begin{equation}
\begin{aligned}
\text{Pr}\{Q\}
&\leq \sum_{i=1}^{m} \text{Pr}\{Z_{n_{1}} \geq (1+\sigma)c_{0}\log n\} \\
&= m n^{-\frac{c_{0}\sigma^{2}}{1+\sigma}} \overset{(a)}{\leq} c n^{1-\frac{c_{0}\sigma^{2}}{1+\sigma}},
\end{aligned}
\label{eq11}
\end{equation}
where $m=\frac{1}{a_{p}^{2}(n)}$ and $c=\pi p$. ($a$) holds using the inequality $\frac{\pi n(1-q)}{\log n}<\pi n(1-q)$. 
 For any sufficiently small $\sigma>0$,  $c_{0}>\frac{1+\sigma}{\sigma^{2}}$ holds.  Hence, $\lim\limits_{n\to \infty}\text{Pr}\{Q\}=0$.
Similarly to (\ref{eq11}), $F$ represents the event that at least one small square contains fewer than $c_{0}(1+\sigma)  \log n$ non-faulty nodes. We have 
\begin{equation}
\begin{aligned}
\text{Pr}\{F\}
&\leq \sum_{i=1}^{m}\text{Pr}\{Z_{n_{1}} \leq (1-\sigma)c_{0}\log n\} \\
&= m n^{-\frac{c_{0}\sigma^{2}}{2}} \leq c n^{1-\frac{c_{0}\sigma^{2}}{2}}.
\end{aligned}
\label{eq13}
\end{equation}

 Let $c_{0}>\frac{2}{\sigma^{2}} >1$,   $\lim\limits_{n\to \infty}\text{Pr}\{F\}=0$. Using (\ref{eq11}) and (\ref{eq13}), we obtain $Z_{n_{1}}=  \Theta \left(\log n\right)$.

\bibliographystyle{IEEEtran}
\bibliography{Reference.bib}

\end{document}